\documentclass[a4paper,10pt]{article}

\usepackage[utf8]{inputenc}

\usepackage{microtype}

\usepackage[numbers]{natbib}
\usepackage{graphicx}

\usepackage{amsthm}
\usepackage{amsmath}
\usepackage{amsopn}
\usepackage{amsfonts}

\usepackage{t1enc}
\usepackage{txfonts}

\usepackage{paralist}
\pltopsep=\smallskipamount

\usepackage[usenames,dvipsnames]{color}
\usepackage{float}
\usepackage{pbox}
\usepackage[nofancy,today]{svninfo}
\svnInfo $Id: simplification.tex 1070 2010-07-26 13:18:38Z ubauer $

\usepackage
[
colorlinks=true,
linkcolor=blue,
anchorcolor=blue,
citecolor=blue,
urlcolor=blue,
]{hyperref}

\newtheorem{theorem}{Theorem}
\newtheorem*{theorem*}{Theorem}
\newtheorem{corollary}[theorem]{Corollary}
\newtheorem{lemma}[theorem]{Lemma}
\newtheorem*{lemma*}{Lemma}

\newtheorem*{proposition*}{Proposition}

\newtheorem*{definition*}{Definition}

\newtheorem*{problem*}{Problem}

\newcommand\RR{{\mathbb{R}}}

\newcommand\NN{{\mathbb{N}}}

\newcommand\Vf{{\ensuremath{V_{\!f}}}}

\DeclareMathOperator\im{im}

\DeclareMathOperator\sd{sd}
\newcommand\tildesigma{{\ensuremath{\tilde\sigma}}}
\newcommand\tildetau{{\ensuremath{\tilde\tau}}}

\newcommand\complex{{\ensuremath{\mathcal{K}}}}
\newcommand\cells{{\ensuremath{K}}}

\definecolor{gray}{rgb}{0.5,0.5,0.5}

\makeatletter
\renewcommand{\@fnsymbol}[1]{\@arabic{#1}}
\makeatother

\title{Optimal topological simplification \goodbreak of discrete functions on surfaces}
\author{
Ulrich Bauer \thanks{Institute for Numerical and
Applied Mathematics,
University of G\"ottingen,
Lotzestr.~16--18,
37083~G\"ottingen, Germany. \texttt{\{bauer,wardetzky\}@math.uni-goettingen.de}
} \and Carsten Lange
\thanks{%
Department of Mathematics and Computer Science,
Freie Universit\"at Berlin,
Arnimallee 6,
14195~Berlin,
Germany. \texttt{lange@math.tu-berlin.de}
} \and Max Wardetzky \footnotemark[1] }

\begin{document}

\maketitle
\begin{abstract}

We solve the problem of minimizing the number of critical points among all functions on a surface within a prescribed distance $\delta$ from a given input function. The result is achieved by establishing a connection between discrete Morse theory and persistent homology. Our method completely removes homological noise with persistence less than $2\delta$, constructively proving the tightness of a lower bound on the number of critical points given by the stability theorem of persistent homology in dimension two for any input function. We also show that an optimal solution can be computed in linear time after persistence pairs have been computed.

\end{abstract}

\section{Introduction}
Measured data and functions constructed from measured data suffer from omnipresent noise introduced during the measuring process. Separating relevant information from noise is therefore a widely considered problem. 

Taking a topological point of view, we regard noise as a source of critical points. Indeed, even arbitrarily small amounts of noise (with respect to the supremum norm) may give rise to an arbitrarily large number of critical points. We may hence interpret critical points that can be \emph{eliminated} by small perturbations as being caused by noise. Consequently, we consider the following optimization problem:

\begin{problem*}[Topological simplification on surfaces]
Given a function $f$ on a surface and a real number $\delta\geq0$, find a function $f_\delta$ subject to $\|f_\delta-f\|_\infty \leq \delta$ such that $f_\delta$ has a minimum number of critical points. 
\end{problem*}

The class of functions and the notion of critical points we work with will be clarified later; for now, we just want to mention that multiple saddles (such as a ``monkey saddle'') are counted here with multiplicity.

The \emph{Bottle\-neck Stability Theorem} \citep{CohenSteiner2007Stability}, a fundamental result in the theory of \emph{persistent homology} \citep{Edelsbrunner2002Topological,Zomorodian2005Computing}, provides a lower bound on the number of critical points:
\begin{proposition*}[Stability Bound]
For any function $f_\delta$ with $\|f_\delta-f\|_\infty \leq\delta$, the number of critical points of $f_\delta$ is bounded from below by the number of critical points of $f$ that have persistence $> 2 \delta$.
\end{proposition*}

Clearly the question about the tightness of this bound is of great importance for the significance of the Bottle\-neck Stability Theorem.
In the present article, we show constructively that the bound given by the stability theorem is actually tight for functions on surfaces (see Theorem~\ref{thm:2deltasimp}):
\begin{theorem*}[Tightness of the stability bound]\label{thm:mainResult}
Given a function $f$ on a surface and a real number $\delta\geq0$, there exists a function $f_\delta$ such that $\|f_\delta-f\|_\infty \leq \delta$ and the number of critical points of $f_\delta$ equals the number of critical points of $f$ that have persistence~$> 2 \delta$.
\end{theorem*}
A similar statement does not hold in higher dimensions or for non-manifold 2-complexes, see Section~\ref{sec:counterexample}.

\subsection{Overview}

\emph{Discrete Morse theory}~\citep{Forman1998Morse,Forman2002Users} provides equivalents of several core concepts of classical Morse theory, like discrete Morse functions, discrete gradient vector fields, critical points, and a cancelation theorem for the elimination of critical points from a vector field. Because of its simplicity, it not only maintains the intuition of the classical theory but allows to go beyond it by providing explicit constructions that would become quite complicated in the smooth setting.

\emph{Persistent homology}~\citep{Edelsbrunner2002Topological,Zomorodian2005Computing} quantifies topological features of a function. It defines the birth and death of homology groups at critical points, identifies pairs of these (\emph{persistence pairs}), and provides a measure of their significance (\emph{persistence}). 

Whereas (discrete) Morse theory makes statements about the \emph{homotopy type} of the sublevel sets of a function, persistence theory is concerned with their \emph{homology}.
Our solution to the problem of topological simplification on surfaces relies on a combination of both theories. In particular, we make contributions to the following problems:

\paragraph{Canceling a single pair of critical points of a function}

\citet{Forman1998Morse} describes a simple method for eliminating pairs of critical points in discrete \emph{vector fields}. Modifying a \emph{function} according to the cancelation of a pair of critical points, however, is more difficult and requires additional effort.
We first observe that a discrete gradient vector field induces a partial order on the cells of the underlying complex, giving rise to the notion of \emph{attracting} and \emph{repelling sets} (in analogy to the notion of stable and unstable manifolds in the classical theory). Building on these concepts,
we describe a canonical method for eliminating a pair of critical points of a discrete Morse \emph{function}. This complements Forman's cancelation method for discrete \emph{gradient vector fields}. In particular, it is applicable in any dimension (Section~\ref{sec:leveling}). An informal description in dimension~1 is shown in Figure~\ref{fig:cancellation}.

\begin{figure*}
\setlength{\tabcolsep}{0pt}%
\begin{tabular*}{\textwidth}{@{\extracolsep{\fill}}ccc}
\includegraphics[width=0.3\textwidth]{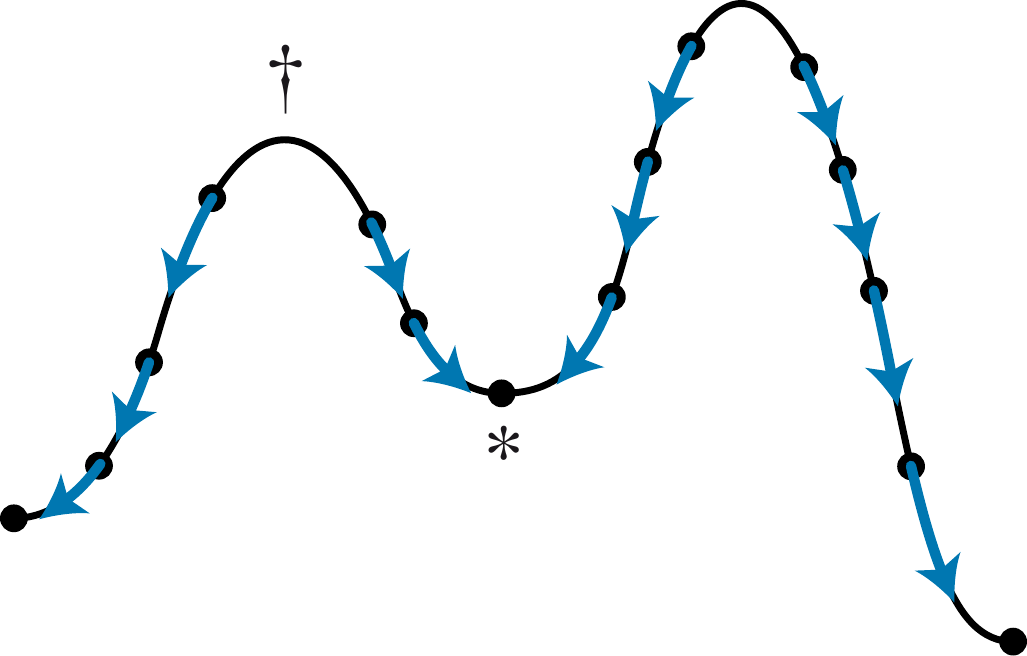}&%
\includegraphics[width=0.3\textwidth]{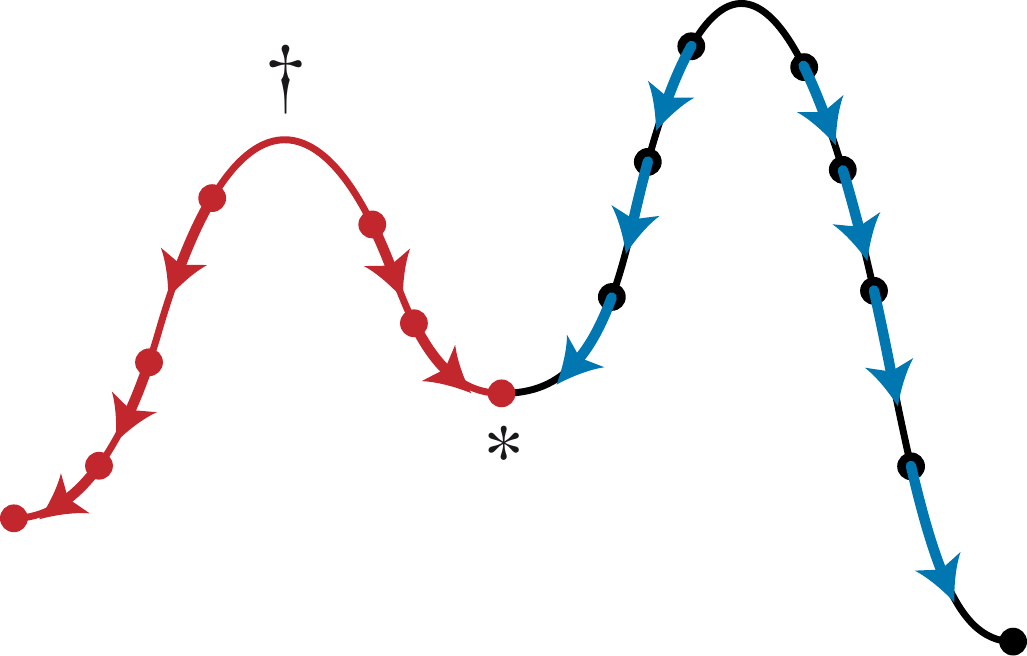}&%
\includegraphics[width=0.3\textwidth]{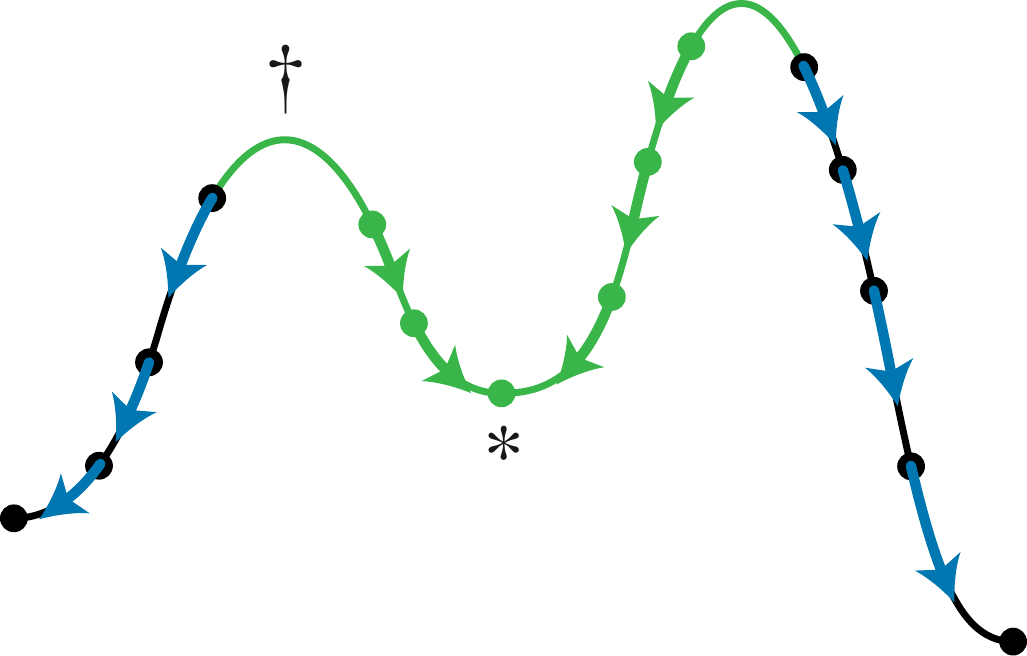}\\
\small(a)&\small(b)&\small(c)\\
\includegraphics[width=0.3\textwidth]{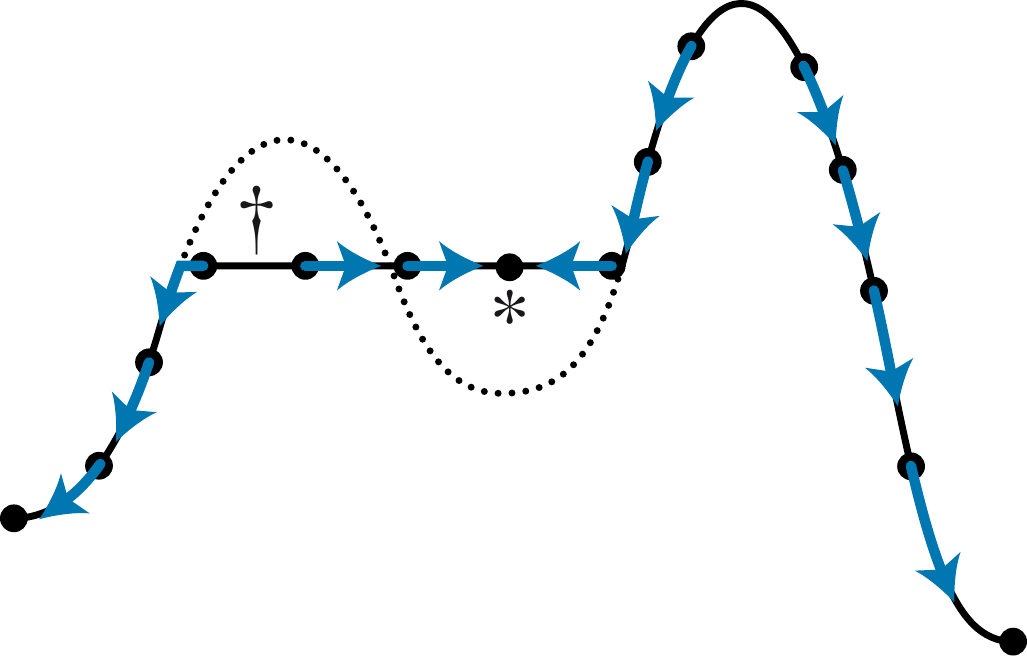}&%
\includegraphics[width=0.3\textwidth]{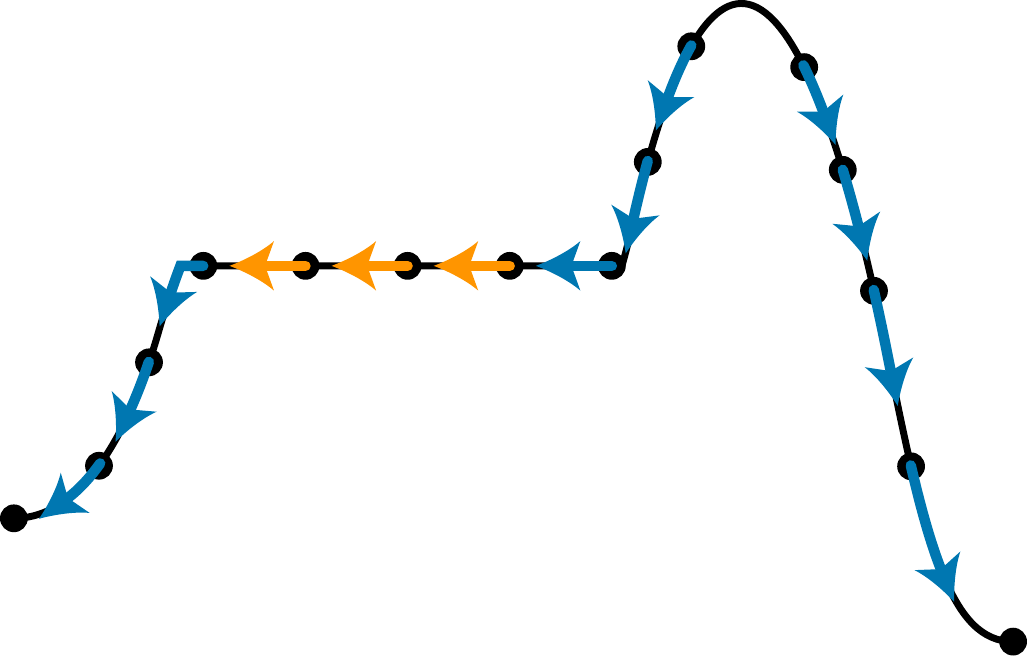}&%
\includegraphics[width=0.3\textwidth]{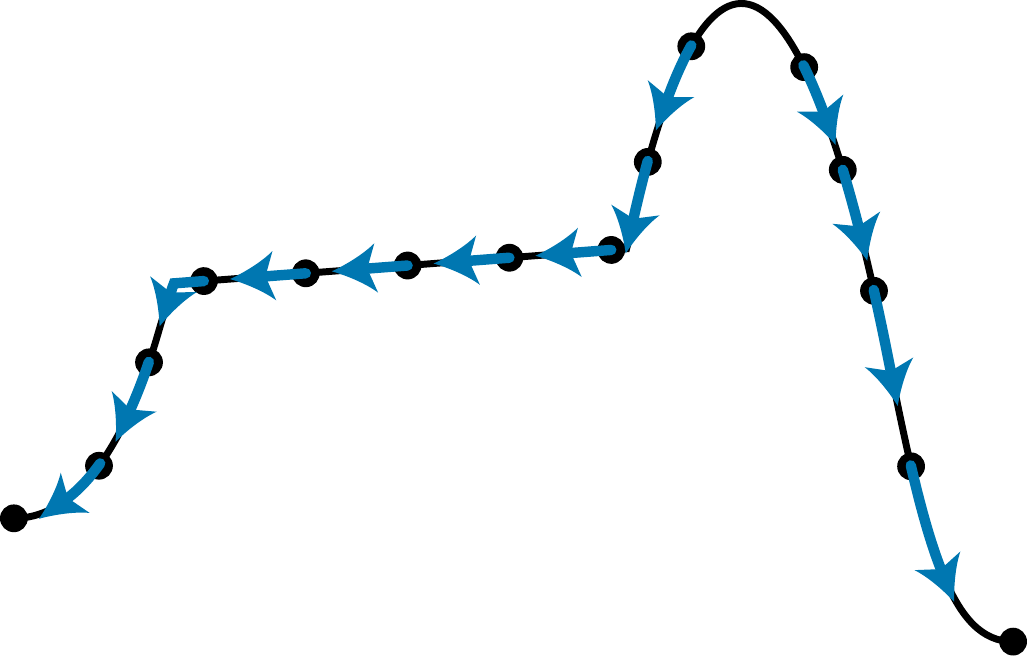}\\
\small(d)&\small(e)&\small(f)\\
\end{tabular*}%
\caption{Cancelation of critical points. (a) shows the graph of a function together with the directions of its gradient vector field. The values of the repelling set (b) of the upper critical point (marked with $\dagger$) and of the attracting set (c) of the lower critical point (marked with $*$) are cut off at the average value of the two critical points, creating a plateau (d). The old gradient directions are still consistent with the new function. The gradient vector field can now be reversed along the path between the critical points, eliminating the pair (e). The resulting function has a plateau, but can be perturbed slightly to become non-degenerate (f).}
\label{fig:cancellation}
\end{figure*}

To cancel a pair of critical points whose values differ by $d$, our method modifies the function by $\frac d 2$ in the supremum norm, which is the minimum required for cancelation (see Figure~\ref{fig:cancellation}). 
To achieve this minimum, elimination of critical points has to take into account the attracting and repelling sets of the canceled pair, containing cells of \emph{all} dimensions; moreover, an arbitrary number of other critical points might have their value changed. This is in contrast to previous related methods \citep{Edelsbrunner2006PersistenceSensitive,Attali2009Simplification} which operate just on the 1-skeleton of the surface (and on the 1-skeleton of its dual) and does not affect other critical points, but modifies the function by $d$ and hence is not a minimal modification in the supremum norm. Moreover, these methods do not extend to higher dimensions.

\paragraph{Degenerate functions}
Morse theory, in any of its variations, fundamentally relies on the assumption that the critical points of the function considered are non-degenerate. This condition not only prevents the theory from being applicable directly to arbitrary input functions.
The canonical function arising from canceling a pair of critical points has a plateau (see Figure~\ref{fig:cancellation}) and hence is not a discrete Morse function. However, there is a Morse function arbitrarily close to it. This necessitates a method to deal with such degenerate functions. To do so, we devise a symbolic perturbation scheme (Section~\ref{sec:pseudoMorse}) based on discrete gradient vector fields, allowing to treat the degenerate case in the same way as the generic case by considering the larger class of \emph{pseudo-Morse functions}; in particular, we do not require the input function to be generic. Instead of deriving information about critical points from the function directly (which leads to ambiguities in degenerate cases), we work with an \emph{explicit} gradient vector field consistent with the function, coinciding with the usual discrete gradient vector field in the generic case. Our scheme always allows to construct a Morse function arbitrarily close to a given pseudo-Morse function and consistent with the given gradient vector field.

A second symbolic perturbation scheme allows to relax the assumption that critical cells have unique function values. It extends the first perturbation scheme by explicitly maintaining a total order on the cells that is consistent with both the function and the gradient vector field.

\paragraph{Multiple cancelations}
We establish a connection between persistence pairs and the cancelation of critical points by proving that for functions on surfaces, \emph{every} persistence pair can eventually be canceled if a sequence of cancelations is performed according to a certain hierarchy on the persistence pairs (Section~\ref{sec:persistenceHierarchy}). The statement is no longer true for manifolds in higher dimensions or non-manifold 2-complexes, where persistence pairs cannot always be canceled.

\paragraph{Tightness of the stability bound}
While the stability bound can easily be seen to be tight when only a single pair of critical points is canceled, we need to ensure that repeated cancelation does not violate the $\delta$-tolerance constraint. Again, the situation is different from previous work \citep{Edelsbrunner2006PersistenceSensitive,Attali2009Simplification}, where simplification is treated separately for pairs of dimensions $(0,1)$ and $(1,2)$. Our result requires to consider cancelation of pairs in different dimensions at the same time (Section~\ref{sec:constraint}).
We provide a constructive proof of the tightness of the stability bound (Theorem~\ref{thm:2deltasimp}). The construction is well suited for proving our theorem; however, it has a suboptimal quadratic time complexity.

\paragraph{Efficient solution}
We show that, after persistence pairs have been computed in time $\mathcal O(\mathop\mathrm{sort}(n))$~\citep{Edelsbrunner2002Topological}, an optimal solution to the topological simplification problem on surfaces can be computed in time $\mathcal O(n)$ using simple graph traversal methods (Section~\ref{sec:efficientAlgo}). Hence, we match the time complexity of \citep{Attali2009Simplification}. 
Since Theorem~\ref{thm:2deltasimp} is already established, we can use it to give a simple proof of correctness of the linear algorithm.

This result is surprising in view of the fact that the topological simplification problem on surfaces is NP-hard when restricting to \emph{simplexwise} linear functions on a triangulated surface%
. This follows from a recent result by~\citet{Gray2010Removing}, which states that minimizing the number of extrema of a simplexwise linear function with interval constraints for the vertex values is NP-hard. Their argument can easily be adapted to our problem setting, where all constraint intervals are assumed to have length $2\delta$. Note that the emphasis on \emph{simplexwise} linear (i.e., linear on each simplex, as opposed to just \emph{piecewise} linear) functions is significant here: a multiple saddle can be split into several non-degenerate saddles by an arbitrarily small (in $L_\infty$) perturbation in the space of piecewise linear functions, but not in the subspace of simplexwise linear functions.

This emphasizes the important role of discrete Morse theory in our problem: the hardness of the problem in the simplexwise linear setting arises from the possibility that the input contains multiple saddles, which is excluded by definition in discrete Morse theory. Going from simplexwise linear functions to discrete Morse functions (Section~\ref{sec:PLfunctions}) can be interpreted as splitting multiple saddles. 

\paragraph{Energy minimization of simplified functions}
The solution to the topological simplification problem is not unique in general: both the $\delta$-constraint and the simplified discrete gradient vector field impose a set of linear inequalities on the simplified function, so the solution set is a convex polytope. This additionally allows to minimize a suitable convex energy functional. We employ this technique to remove artifacts from the initial solution and to improve the similarity to the input function (Section~\ref{sec:energyMethods}).

\subsection{Related work}
Topological simplification of functions within a $\delta$-tolerance constraint has been considered before by
\citet{Edelsbrunner2006PersistenceSensitive} and \citet{Attali2009Simplification}. The problem considered there differs from ours by a seemingly small but significant detail: in \citep{Edelsbrunner2006PersistenceSensitive, Attali2009Simplification} the critical points of the input function $f$ that are not eliminated are additionally assumed to retain the same critical value for the output $g$. This restriction has serious consequences: while it allows for the elimination of critical points of $f$  with persistence $\leq \delta$, in certain cases not all critical points with persistence $\leq 2\delta$ can be eliminated; an example is given in~\citep{Edelsbrunner2006PersistenceSensitive}. Hence, under this restriction it is not possible to match the stability bound. Moreover, this result does not provide any information about the tightness of the stability bound since it considers only a restricted set of functions.

The methods presented in \citep{Edelsbrunner2006PersistenceSensitive, Attali2009Simplification} can be interpreted as variants of the \emph{carving} method proposed by~\citet{Soille2004Morphological} in the context of terrain simplification. There is another popular method for removing extrema from terrains, called \emph{filling} or \emph{flooding} \citep{Jenson1988Extracting,Agarwal2006IOefficient,Danner2007TerraStream}. A combination of both methods has been proposed in \citep{Soille2004Optimal}.
Our methods of canceling critical points from a function can be interpreted as a combination of carving and flooding in the realm of discrete Morse theory.

Apart from the above mentioned works, persistent homology provides the basis of several other elegant methods for computation and simplification of multi-scale structures derived from a function. For example, \citet{Edelsbrunner2002Topological} discuss simplification of the persistent homology for \emph{filtrations} of simplicial complexes. \citet{Edelsbrunner2003Hierarchical} and \citet{Gyulassy2007Efficient} consider simplification of \emph{cell decompositions} (Morse-Smale complexes) resulting from a given Morse function. Unfortunately, a simplified Morse-Smale complex does not directly give rise to a simplified function. Indeed, simplifying a Morse-Smale complex is closely related to simplifying a discrete gradient vector field.

The problem of constructing discrete \emph{gradient vector fields} (as opposed to functions) that minimize the number of critical points \emph{without} constraints is addressed by \citet{Lewiner2003Optimal} for surfaces and by \citet{Joswig2006Computing} for complexes of arbitrary dimension. \citet{King2005Generating} were the first to propose the combination of persistence with discrete Morse theory to simplify the gradient vector field of an input function on a 3-dimensional simplicial complex. Their method has quadratic time complexity and produces a  simplified discrete gradient vector field but not a function. Moreover, it does not aim at optimality (in 3 dimensions, not every persistence pair can be canceled).

Several statements of this article can also be transferred to the setting of piecewise linear Morse-Smale complexes. For example, Theorem~\ref{thm:simplifiedVectorField} can be used to show that the successive simplification of a Morse-Smale complex on a surface proposed by~\citet{Edelsbrunner2003Hierarchical} is always possible. This extends the Adjacency Lemma in 
\citep{Edelsbrunner2003Hierarchical}, which shows a necessary but not sufficient condition for the successive cancelation of persistence pairs.

\section{Discrete Morse theory}\label{sec:MorseTheory}

Classical (smooth) Morse theory~\citep{Milnor1963Morse} relates the critical points of a generic smooth real-valued function on a manifold to the global topology of that manifold. \citet{Forman1998Morse,Forman2002Users} carried over the main ideas of Morse theory to a combinatorial setting. We briefly review some important notions and results here that are used throughout this article, together with some extensions to Forman's theory that provide important tools for our solution.

A CW complex $\complex$ is a topological space constructed inductively: starting with a discrete set $\complex_0$ of 0-cells, the $n$-skeleton $\complex_n$ is formed by attaching $n$-cells (open $n$-dimensional balls) by continuous maps $\mathbb{S}^{n-1}\to \complex^{n-1}$ from their boundary to the $(n-1)$-skeleton. The set of cells of $\complex$ is denoted by~$\cells$. Throughout this article, we consider only finite CW complexes. Whenever a cell $\tau \in \cells$ is attached to a cell~$\sigma$ (i.e., 
$\sigma \subset \partial\tau$,
where $\partial\tau$ denotes the boundary of $\tau$), we call $\sigma$ a \emph{face} of $\tau$; a face of codimension 1 is called a \emph{facet}. 
If all attaching maps are homeomorphisms, $\complex$ is called a \emph{regular} CW complex. 
A regular CW complex whose underlying space is a 2-manifold is called a \emph{combinatorial surface}. We refer to \citep{Lundell1969Topology,Hatcher2002Algebraic} for details on CW complexes.

\subsection{Discrete vector fields}
One of the central concepts of discrete Morse theory is that of a \emph{discrete vector field} -- a purely combinatorial analogue of a classical vector field.
\begin{definition*}[discrete vector field, critical cell \citep{Forman1998Morse,Forman2002Users}]
A \emph{discrete vector field} $V$ on a regular CW complex $\complex$ is a set of pairs of cells $(\sigma,\tau)\in \cells\times\cells$, with $\sigma$ a facet of $\tau$, such that each cell of~$\cells$ is contained in at most one pair of~$V$. A cell $\sigma\in\cells$ is \emph{critical} with respect to $V$ if $\sigma$ is not contained in any pair of $V$. The dimension of a critical cell is also called its index.
\end{definition*}

A pair $(\sigma,\tau)$ in a discrete vector field $V$ can be visualized as an arrow from~$\sigma$ to~$\tau$ (as in Figure~\ref{fig:reverse_a_v_path}).
\begin{figure}
\center{
\includegraphics[scale=0.9]{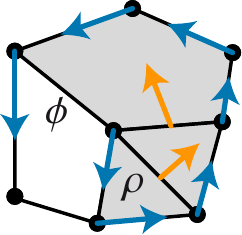}
\qquad
\includegraphics[scale=0.9]{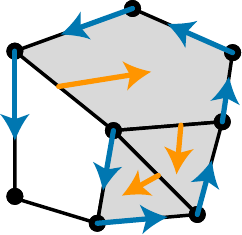}
}
\caption{Reversing a gradient vector field along the unique path from $\partial\rho$ to $\phi$ produces a gradient vector field in which the 1-cell $\phi$ and the 2-cell $\rho$ are no longer critical.}
\label{fig:reverse_a_v_path}
\end{figure}

In the following, we consider an important subclass of vector fields in which the arrows do not form closed paths. This can be made precise using the concept of \emph{$V$-paths}. 
\begin{definition*}[$V$-path \citep{Forman2002Users}]
Let~$V$ be a discrete vector field. A
$V$-path~$\Gamma$ from a cell $\sigma_0$ to a cell $\sigma_r$ is a sequence
$(\sigma_0,\tau_0, \sigma_1, \ldots ,\tau_{r-1}, \sigma_r)$ of cells such that for
every $0\leq i\leq r-1$:
$$
\sigma_i\,\text{ is a facet of }\,\tau_i  ~\,\text{with}~\, (\sigma_i,\tau_i)\in V \quad\text{and}\quad
\sigma_{i+1}\text{ is a facet of }\,\tau_i  ~\, \text{with}~\, (\sigma_{i+1},\tau_i)\not\in V.
$$
$\Gamma$ is \emph{closed} if $\sigma_0=\sigma_r$ and \emph{nontrivial} if $r>0$. We call $\dim\sigma_0$ the \emph{dimension of $\Gamma$}.
\end{definition*}
By a $V$-path from~$\partial\rho$ to~$\phi$ we mean a $V$-path from a facet of $\rho$ to $\phi$ (see Figure~\ref{fig:reverse_a_v_path} for an example).

\begin{definition*}[discrete gradient vector field \citep{Forman2002Users}]$ $
A discrete vector field $V$ is a \emph{discrete gradient vector field} if it contains no nontrivial closed $V$-paths. 
\end{definition*}

The main technique
for reducing the number of critical points is that of \emph{reversing}
a gradient vector field~$V$ along a~$V$-path between two critical cells~$\rho$
and~$\phi$ (see Figure~\ref{fig:reverse_a_v_path} for an example). It provides a discrete analogue of Morse's cancelation theorem~\citep{Morse1964Elimination}:%

\begin{theorem}[\citet{Forman1998Morse}]\label{thm:cancelation}
Let~$\phi$ and $\rho$ be two critical cells of a gradient vector field~$V$ with exactly one $V$-path~$\Gamma$ from~$\partial\rho$ to~$\phi$. Then there is a gradient vector field~$\widetilde V$ obtained by reversing $V$ along the path~$\Gamma$. The critical cells of $\widetilde V$ are exactly the critical cells of~$V$ apart from $\left\{ \phi,\rho \right\}$. Moreover, $V = \widetilde V$ except along the path~$\Gamma$.
\end{theorem}

Gradient vector fields on combinatorial surfaces have additional properties that do not hold in higher dimensions. The following property is readily checked using the fact that a 1-cell is only attached to at most two 0-cells, and at most two 2-cells are attached to a 1-cell:

\begin{lemma}\label{lem:uniquePredSucc}
Two $V$-paths of dimension 0 cannot branch at a common cell, and two $V$-paths of dimension 1 cannot merge (except at their last cell).
\end{lemma}

\begin{corollary}\label{cor:twoPaths}
Let $\rho$ be a critical 1-cell of a discrete vector field $V$ on a
combinatorial surface. Then there are at most two $V$-paths from $\partial\rho$ to critical 0-cells, each starting at one of the two 0-cells in $\partial\rho$. Similarly, there are at most two $V$-paths from facets of critical 2-cells to $\rho$.
\end{corollary}

\subsection{Pseudo-Morse functions and symbolic perturbation}\label{sec:pseudoMorse}

As in smooth Morse theory, a discrete gradient vector field can be understood as the gradient of some non-degenerate function in the following sense:

\begin{definition*}[discrete Morse function \citep{Forman1998Morse}]
A function $f:\cells \rightarrow \RR$ on the cells of a regular CW complex $\complex$ is a \emph{discrete Morse function} if there is a gradient vector field~$\Vf$ such that whenever $\sigma $ is a facet of $\tau$ then
$$
(\sigma,\tau) \not\in \Vf ~\, \text{implies} ~\, f(\sigma) < f(\tau) \quad\text{and}\quad
(\sigma,\tau) \in \Vf ~\, \text{implies} ~\, f(\sigma) \geq f(\tau).   
$$
$\Vf$ is called the \emph{gradient vector field of $f$}.
\end{definition*}
In contrast to simplexwise linear functions, which are determined by their function values at the vertices, discrete Morse functions take values on cells of \emph{any} dimension.

The gradient vector field of a discrete Morse function encodes only the \emph{sign} of the difference between function values, not the difference itself. Therefore a discrete gradient vector field does not uniquely determine a discrete Morse function, but for every discrete Morse function $f$ there is exactly one gradient vector field~$\Vf$.

In order to be able to treat non-generic input functions, it is useful
to consider a more general class of functions, which we call \emph{pseudo-Morse functions}.
Pseudo-Morse functions substitute the strict inequality in the definition of
Morse functions by a weak one.

\begin{definition*}[pseudo-Morse function, consistency]
A function $f:\cells \rightarrow \RR$ on the cells of a regular CW complex $\complex$ is a \emph{discrete pseudo-Morse function} if there is a gradient vector field $V$ such that whenever $\sigma $ is a facet of $\tau$ then
$$
(\sigma,\tau) \not\in V ~\, \text{implies} ~\, f(\sigma) \leq f(\tau) \quad\text{and}\quad
(\sigma,\tau) \in V ~\, \text{implies} ~\, f(\sigma) \geq f(\tau).   
$$
In this case, we call $f$ and $V$ \emph{consistent}.
\end{definition*}

Note that a gradient vector field $V$ consistent with a pseudo-Morse function $f$ is not unique in general. The following lemma provides
a useful characterization of discrete pseudo-Morse functions.
\begin{lemma}\label{lem:perturbation}
Let $f:\cells\to\RR$ be a function on the cells of a regular CW complex $\complex$ and let $V$ be a gradient vector field on $\complex$. Then $f$ is a discrete pseudo-Morse function consistent with $V$ if and only if for every $\epsilon>0$ there is a discrete Morse function $f_\epsilon:\cells\to\RR$ with $\|f_\epsilon-f\|_\infty \leq \epsilon$ such that $V$ is the gradient vector field of $f_\epsilon$. 
\end{lemma}

\begin{proof}
Assume that $f$ is a pseudo-Morse function consistent with a gradient vector field $V$. %
There exists a discrete Morse function $g$ whose gradient vector field $V_g$ is precisely given by $V$~\citep{Forman1998Morse}. Let $G$ be the maximum absolute value of~$g$. Given $\epsilon >0$, for each cell $\sigma$ define $f_\epsilon(\sigma):=f(\sigma)+\epsilon \frac{g(\sigma)}{G}$. Then it is straightforward to check that $f_\epsilon$ is a discrete Morse function with gradient vector field $V$ and $\|f_\epsilon-f\|_\infty\leq \epsilon$. 

On the other hand, assume that for every $\epsilon>0$ there is a discrete Morse function $f_\epsilon:\cells\to\RR$ consistent with $V$ and $\|f_\epsilon-f\|_\infty \leq \epsilon$. Choose $\epsilon$ such that for every $\phi, \rho \in K$ with $f(\phi)\neq f(\rho)$ we have $\epsilon < \frac{|f(\phi)-f(\rho)|}2$. In this case, one easily verifies that $f$ is a pseudo-Morse function consistent with~$V$.
\end{proof}

The previous lemma provides a \emph{symbolic perturbation} scheme based on gradient vector fields in order to allow for non-generic (degenerate) input functions. Starting with a pseudo-Morse function $f$, we can choose a consistent gradient vector field $V$, which may not be unique. Lemma~\ref{lem:perturbation} asserts that there is  a discrete Morse functions~$f_\epsilon$ arbitrarily close to $f$ and consistent with $V$. Therefore we can work with $f$ as if it were a discrete Morse function with gradient vector field $V$. In particular, we use it to consider critical points associated to a pseudo-Morse function by choosing a consistent gradient vector field.

This first symbolic perturbation scheme is not sufficient for all our purposes; the definition of persistence pairs in Section~\ref{sec:persistenceMorse} not only requires a gradient vector field, but also a total order on the critical cells, which again might not be uniquely defined by a pseudo-Morse function $f$ and a consistent gradient vector field $V$. We now derive a second perturbation scheme that meets these requirements.

Since a gradient vector field imposes certain inequality constraints on the functions consistent with it, we can ask how these inequalities affect the relation between the function values of any two cells. We observe that any discrete gradient vector field gives rise to a strict partial order on the set of cells:

\begin{definition*}[induced partial order]
Let $V$ be a discrete gradient vector field and consider the
relation~$\gets_V$ defined on $K$ such that whenever $\sigma $ is a facet of $\tau$ then
$$
(\sigma,\tau) \not\in V ~\, \text{implies} ~\, \sigma \gets_V \tau \quad\text{and}\quad
(\sigma,\tau) \in V ~\, \text{implies} ~\, \sigma \to_V \tau.   
$$
Let $\prec_V$ be the transitive closure of $\gets_V$. Then $\prec_V$ is called the \emph{(strict) partial order induced by} $V$.
\end{definition*}
The interpretation of this partial order is that for \emph{any} pseudo-Morse function $f$ consistent with $V$ and any two cells $\phi$ and $\rho$, the relation $\phi \prec_V \rho$ implies $f(\phi)\leq f(\rho)$. The relation~$\gets_V$ is the \emph{covering relation} of~$\prec_V$, i.e., $\phi \gets_V \rho$ implies $\phi \prec_V \rho$ and there is no $\psi$ with $\phi \prec_V \psi \prec_V \rho$. The covering relation of a partial order forms a directed acyclic graph called the \emph{Hasse diagram} (with edges oriented as suggested by~$\gets_V$). The Hasse diagram~$H_V$ of~$\prec_V$ is obtained from the Hasse diagram of the face lattice of $\complex$ by inverting the orientation of all edges corresponding to pairs $(\sigma,\tau)\in V$ as described by \citet{Chari2000Discrete}. $H_V$ has the property that $\phi \prec_V \rho$ if and only if there is a directed path from $\rho$ to $\phi$. Note that $\sigma \gets_V \tau$ implies $f(\sigma) \leq f(\tau)$, i.e., both the arrow visualizing $(\sigma,\tau)\in V$ and the arrow symbolizing $\sigma \gets_V \tau$ point towards a (weakly) \emph{decreasing} function value of $f$. 

Assume we are given a pseudo-Morse function $f$ consistent with a gradient vector field~$V$. On the one hand we have the induced partial order~$\prec_V$. On the other hand the function $f$ canonically induces a strict partial order~$\prec_f$ given by $\phi \prec_f \rho \Leftrightarrow f(\phi)<f(\rho)$. 
Since the two orders~$\prec_f$ and~$\prec_V$ are compatible by assumption (there are no two cells $(\phi,\rho)$ with $\phi\prec_V\rho$ and $\phi \succ_f \rho$), we can merge them into a strict partial order $\prec_{f,V}$ (the transitive closure of $(\prec_f \cup \prec_V)\subset \cells\times \cells$). A linear extension of this order is now a strict total order~$\prec$ consistent with both $f$ and $V$.

\begin{definition*}[consistent total order]
Let $V$ be a discrete gradient vector field $V$ consistent with a discrete pseudo-Morse function $f$. Then a strict total order $\prec$ is called \emph{consistent} with $(f,V)$ if it is a linear extension of $\prec_f$ and $\prec_V$. %
\end{definition*}

This total order~$\prec$ gives rise to a canonical function $\cells\to\NN$, which is a discrete Morse function and consistent with $V$. If we use this function as the function $g$ in the proof of Lemma~\ref{lem:perturbation} to construct $f_\epsilon$, then $f_\epsilon$ is an injective discrete Morse function with gradient vector field $V$ and the total order induced by $f_\epsilon$ is~$\prec$ again. We thus obtain a second symbolic perturbation scheme for situations where a total order on the cells is required.

We make use of this concept in the following definition. A classical object of study in smooth Morse theory is the \emph{sublevel set} $\{x\in M: f(x) \leq t\}$ of a function $f:M\to\RR$ on a manifold $M$. In the discrete theory, the analogous object is the \emph{level subcomplex}, and the equivalent construction using our second symbolic perturbation scheme is the \emph{order subcomplex}:

\begin{definition*}[level subcomplex \citep{Forman1998Morse}, order subcomplex]
Let $f$ be a pseudo-Morse function on a regular CW complex $\complex$. Let the \emph{carrier} of a subset $L \subset \cells$ be the smallest subcomplex of $\complex$ containing all of~$L$. Then for $t \in \RR$, the \emph{level subcomplex} is
$$
\complex(t)=\mathop{\mathrm{carrier}}\Bigg(
\bigcup_{\rho \in \cells:f(\rho) \leq t} \rho \Bigg).
$$
Similarly, let $\prec$ be a strict total order on the cells~$\cells$ of
a regular CW complex $\complex$. 
Then for a cell $\sigma \in \cells$, the \emph{order subcomplex} is
$$
\complex(\sigma)=\mathop{\mathrm{carrier}}\Bigg(
\mathop{\bigcup_{\rho \in \cells:\rho \preceq \sigma}} \rho \Bigg).
$$
\end{definition*}

Like in the smooth theory, the homotopy type of level subcomplexes changes only at critical cells. The statement can trivially be rephrased for order subcomplexes:
\begin{theorem}[\citet{Forman1998Morse}]%
\label{thm:subcomplexesHomotopyEquiv}
Let $V$ be a gradient vector field on $\complex$ and let $\prec$ be a linear extension of $\prec_V$.
If $\rho$ and $\psi$ are two cells such that $\rho\prec\psi$ and there is no critical cell $\phi$ with respect to $V$ such that $\rho \prec \phi \preceq \psi$, then $\complex(\psi)$ collapses to $\complex(\rho)$.
\end{theorem}

The order subcomplexes provide a finer (cell-by-cell) filtration of the complex $\complex$ than the level subcomplexes, in particular if $f$ is degenerate. This turns out to be useful when working with persistent homology in Section~\ref{sec:persistenceMorse}.

\subsection{Piecewise linear functions and discrete Morse functions}\label{sec:PLfunctions}
In this section we discuss a canonical relationship between discrete and piecewise linear (PL) Morse theory. As it turns out, it is possible to translate statements from one setting to the other seamlessly. %
Similar constructions have been used by~\citet{King2005Generating, Attali2009Simplification}. 

Assume that $\complex$ is a simplicial complex. 
Let $f_\text{PL}$ be a simplexwise linear function on $\complex$ and let $f_0$ be its restriction to the 0-skeleton of $\complex$.
The function $f_0$ inductively gives rise to a discrete pseudo-Morse function $f$ in the following way. For each 0-cell~$\alpha$, let $f(\alpha)=f_0(\alpha)$. For a cell $\tau$ with $\dim\tau>0$, let $f(\tau)$ be the maximum value of $f$ on any facet of~$\tau$. The function $f$ can easily be seen to be pseudo-Morse since it is consistent with the empty vector field $V=\emptyset$  (all cells are critical). Note that any level subcomplex of $f$ coincides with the induced subcomplex of $\complex$ on the corresponding sublevel set of~$f_0$. This induced subcomplex, in turn, is homotopy equivalent to the corresponding sublevel set of $f_\text{PL}$ 
\citep{Kuhnel1990Triangulations,Morozov2008Homological}. This means that from a Morse-theoretic point of view, the PL function $f_\text{PL}$ and the pseudo-Morse function $f$ are equivalent. We conclude:

\begin{theorem}
Let $f_\text{PL}$ be a simplexwise linear function on 
a simplicial complex $\complex$. Then there is a canonical pseudo-Morse function $f$ on $\complex$ such that for every $t\in\RR$ the sublevel set $\{x\in \complex: f_\text{PL}(x) \leq t\}$ is homotopy equivalent to the level subcomplex $\complex(t)$.
\end{theorem}

Vice versa, we can interpret any discrete pseudo-Morse function $f$ on a regular CW complex $\complex$ as a simplexwise linear function $f_\text{sd}: |\!\sd\complex| \to \RR$ on the underlying space of the barycentric subdivision $\sd\complex$.
The \emph{barycentric subdivision} of a regular CW complex $\complex$ is the order complex of the face lattice, i.e., the \emph{abstract} simplicial complex $\sd\complex$ whose vertices are the cells of $\complex$ and whose simplices are the totally ordered subsets of the face lattice.
The underlying space $|\!\sd\complex|$ is homeomorphic to~$\complex$.
The function $f_\text{sd}$ is assumed to linearly interpolate the values of $f$ at the vertices of $|\!\sd\complex|$ inside each simplex of $|\!\sd(\complex)|$. 
Again, the sublevel sets of $f_\text{sd}$ are homotopy equivalent to the corresponding level subcomplexes of $f$: 

\begin{figure}[ht]\label{fig:morse_PL_equivalence}
\center{
\includegraphics[scale=0.9]{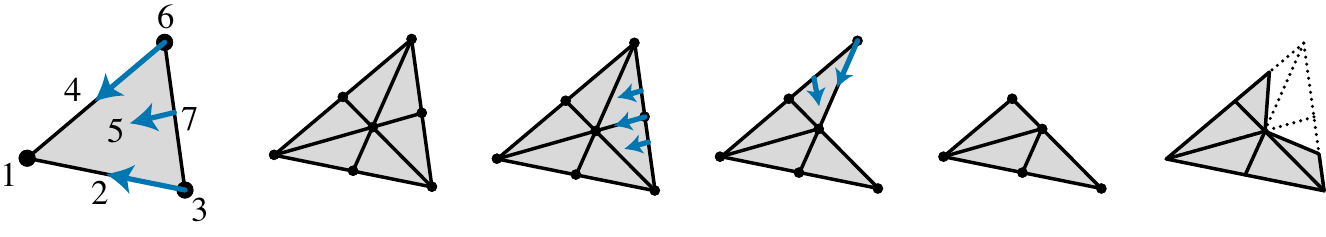}%
}
\caption{Illustration to Theorem~\ref{thm:morse_PL_equivalence}, showing the homotopy equivalence of the level subcomplex $\complex(t)$ to the sublevel set $\{x\in |\!\sd\complex|: f_\text{sd}(x) \leq \nobreak t\}$ for $t=5$. From left to right: function $f$ on $\complex(5)$; barycentric subdivision $\sd\complex(5)=\Delta(K(5))=\Delta(7)$; vector fields defining the collapse of $\Delta(7)$ onto $\Delta(6)$ and of $\Delta(6)$ onto $\Delta(5)$; sublevel set of~$f_\text{sd}$.}
\end{figure}

\begin{theorem}\label{thm:morse_PL_equivalence}
Let $f$ be a pseudo-Morse function on a simplicial complex $\complex$. Then $f$ induces a simplexwise linear function $f_\text{sd}$ on $|\!\sd\complex|$ such that for every $t\in\RR$ the level subcomplex $\complex(t)$ is homotopy equivalent to the sublevel set $\{x\in |\!\sd\complex|: f_\text{sd}(x) \leq \nobreak t\}$.
\end{theorem}

\begin{proof}

Let $V$ be a discrete gradient vector field on $\complex$ that is consistent with $f$ and let $\prec$ be a total order consistent with $(f,V)$. 
Let $\cells(t)$ and $\cells(\rho)$ denote the cells of the level and order subcomplexes $\complex(t)$ and $\complex(\rho)$, respectively. 
Let $\Delta(U)$ denote the induced subcomplex of $\sd\complex$ on a vertex set $U$ (we identify a cell $\rho\in K$ with the corresponding vertex $\{\rho\}\in\sd\complex$). 
The induced subcomplex $\Delta(K(t))$ is easily seen to be identical to $\sd\complex(t)$.
Let $F(t)=\{\phi \in \cells: f(\phi) \leq t\}\subset \cells(t)$. We now show that $\Delta(K(t))$ collapses simplicially onto $\Delta(F(t))$. See Figure~\ref{fig:morse_PL_equivalence} for an example.

Let $\sigma \in K(t)\setminus F(t)$ and let $\sigma_-$ denote its predecessor with respect to $\prec$. We write $\Delta(\rho)$ for $\Delta(\{\phi\in K:\phi\prec\rho\})$. We show that $\Delta(\sigma)$ collapses onto $\Delta(\sigma_-)$.
It follows from the definition of an order subcomplex that $(\sigma,\tau) \in V$ for a unique $\tau\in F(t)$. %
Consequently, for every simplex $S\in\Delta(\sigma)$ with $\sigma\in S$ and $\tau \not\in S$ the simplex $T=S \cup \{\tau\}$ is also contained in $\Delta(\sigma)$%
. Hence, these pairs $(S,T)$ constitute a discrete gradient vector field~$W$ on~$\Delta(\sigma)$ such that exactly the simplices containing $\sigma$ (the \emph{vertex star} of $\sigma$) are non-critical. This vector field~$W$ provides a simplicial collapse of $\Delta(\sigma)$ onto $\Delta(\sigma_-)$ by applying Theorem~\ref{thm:subcomplexesHomotopyEquiv} with an arbitrary linear extension of $\prec_W$.
By repeatedly applying this argument, we obtain that $\Delta(K(t))$ collapses onto $\Delta(F(t))$. This implies that the underlying spaces are homotopy equivalent.

Finally, let $f_\text{sd}$ be the simplexwise linear extension of $f$ from the vertices of $\sd(\complex)$ to the whole complex.
Recall that the induced subcomplex $\Delta(F(t))$ is homotopy equi\-va\-lent to the sublevel set $\{x\in |\!\sd\complex|: f_\text{sd}(x) \leq \nobreak t\}$ \citep{Kuhnel1990Triangulations,Morozov2008Homological}. The claim now follows.
\end{proof}

This equivalence allows us to translate back and forth between piecewise linear functions and pseudo-Morse functions, and to apply theorems of piecewise linear Morse theory to  discrete Morse theory. 

In a similar fashion, a discrete pseudo-Morse function can be constructed from a function defined only on the 2-cells of a combinatorial surface by defining $f(\sigma)$ as the minimum value of all cells that contain $\sigma$ as a facet. This can be used to construct discrete pseudo-Morse functions from functions defined on cubical grids, such as pixel images, by interpreting each pixel as a 2-cell. The resulting level subcomplexes correspond to the cubical complexes extracted from images as described by \citet{Kaczynski2004Computational}. Vice versa, a pseudo-Morse function on a cubical complex can be interpreted as a function defined on a subdivided grid. This construction has been used in the examples in Section~\ref{sec:discussion}.

Note that starting with a PL function and constructing a pseudo-Morse function consistent with the empty vector field means that initially all cells are considered critical, which is a point worth discussing. \citet{King2005Generating} propose to construct an initial discrete gradient vector field with critical cells corresponding to the critical vertices (in the PL sense, see \citep{Kosinski1962Singularities,Eells1962Manifolds,Banchoff1967Critical}) of a (non-degenerate) input PL function instead. We omit such a step for two reasons. First, this step is unnecessary in our method and would not lead to different results. Second, the step can actually be interpreted as a special case of the topological simplification problem with $\delta=0$. In this case, the problem reduces to minimizing the number of critical points among all gradient vector fields consistent with the input function. We discuss the simplification of a gradient vector field in Sections~\ref{sec:persistenceHierarchy} and~\ref{sec:efficientAlgoVF}. 

\section{Persistent homology of discrete Morse functions}\label{sec:persistenceMorse}

The notions of persistent homology and persistence pairs were introduced in \citep{Edelsbrunner2002Topological,Zomorodian2005Computing,CohenSteiner2007Stability} in order to investigate the change of the homology groups in a filtration of a topological space (a nested sequence of subspaces). This concept can naturally be applied to discrete pseudo-Morse functions.
The following definitions can be applied to cellular homology with coefficients in an arbitrary field $F$.
We write
$H_d(\complex)$ as a shorthand for the $d^\mathrm{th}$ homology group
$H_d(\complex;F)$ of $\complex$ and $H_*(\complex) = \bigoplus_{d} H_d(\complex)$. %

\paragraph{Convention and Notation}
Throughout Section~\ref{sec:persistenceMorse} we consider a pseudo-Morse function $f$ consistent with a gradient vector field $V$ on a regular CW complex $\complex$ and a strict total order~$\prec$ consistent with~$(f,V)$.

\subsection{Birth, death, and persistence pairs}\label{sec:persistence}

As a consequence of Theorem~\ref{thm:subcomplexesHomotopyEquiv}, the homology groups of order subcomplexes change only at critical cells of $V$.
Let $\sigma$ and $\tau$ be critical cells such that $\sigma \prec \tau$ and consider the inclusion map $i^{\,\sigma,\,\tau}:\complex(\sigma)\hookrightarrow \complex(\tau)$ between the order subcomplexes with regard to the total order $\prec$. This map induces a homomorphism
\ensuremath{
i_*^{\,\sigma,\,\tau}:H_*(\complex(\sigma))\rightarrow H_*(\complex(\tau))
}
between homology groups.
For every cell $\rho$, let $\rho_-$ denote its predecessor with regard to~$\prec$. 
Now consider the sequence
$$
H_*(\complex(\sigma_{-})) \to
H_*(\complex(\sigma)) \to
H_*(\complex(\tau_{-})) \to
H_*(\complex(\tau))
$$
of induced homomorphisms. Here we
allow for the cases $\sigma=\tau_-$ and $\sigma_-=\emptyset$ (if $\sigma$
is the first cell in $\prec$, in which case $H_*(\complex(\sigma_{-}))$ is the trivial group).

\begin{definition*}[birth, death, persistence pair \citep{Edelsbrunner2002Topological}]%
We say that a
class $h\in H_*(\complex(\sigma))$ is \emph{born at} (or \emph{created by}) a positive
cell $\sigma$ if 
$$
h\not\in \im(i_*^{\,\sigma_-,\,\sigma}).
$$
Moreover, we
say that a class $h\in H_*(\complex(\sigma))$ that is born at $\sigma$ \emph{dies
entering} (or gets \emph{merged by}) a negative cell $\tau$ if there is a
class $\tilde h\in H_*(\complex(\sigma_-))$ such that
$$
i_*^{\,\sigma,\,\tau_-}(h)\not\in \im(i_*^{\,\sigma_-,\,\tau_-}) %
\quad\text{but}\quad%
i_*^{\,\sigma,\,\tau}(h) = i^{\,\sigma_-,\,\tau}_*(\tilde h) \in \im(i_*^{\,\sigma_-,\,\tau}).
$$
If there
exists a class $h$ that is born at $\sigma$ and dies entering $\tau$, then
$(\sigma, \tau)$ is a \emph{persistence pair}. The difference $f(\tau)-f(\sigma)$ is called the \emph{persistence} of $(\sigma, \tau)$. %
\end{definition*}
Note that in this definition we always have $\dim\tau=\dim\sigma+1$. On combinatorial surfaces, the only possible cases for $(\dim\sigma,\dim\tau)$ are $(0,1)$ or $(1,2)$.

\subsection{Duality and persistence}

For any closed combinatorial surface $\complex$, there is an associated \emph{dual complex} $\complex^*$, a combinatorial surface homeomorphic to~$\complex$ whose $i$-cells correspond to $(2-i)$-cells of $\complex$~\citep{Hatcher2002Algebraic}.
A discrete pseudo-Morse function $f$ on $\complex$ gives rise to a discrete pseudo-Morse function $f^*$ on $\cells^*$ via $\sigma^*\mapsto-f(\sigma)$ \citep{Forman1998Morse}.

Moreover, as shown by \citet{CohenSteiner2008Extending} and \citet{Attali2009Simplification}, the persistence pairs of dimension $(1,2)$ for $\complex$ correspond to the persistence pairs of dimension $(0,1)$ for the dual complex $\complex^*$ (with $\tau^*\prec\sigma^* \Leftrightarrow \sigma\prec\tau$). 
The 
homology groups $H_0(\complex(\rho_i))$ (generated by the connected components of $\complex(\rho_i)$), and hence the 
persistence pairs of dimension $(0,1)$, are determined solely by the 1-skeleton of $\complex$, also called the \emph{(primal) graph} of $\complex$. Consequently, the persistence pairs of dimension $(1,2)$ are determined by the 1-skeleton of $\complex^*$, called the \emph{dual graph}. 
This means that all persistence pairs of a surface can be determined in terms of Morse functions on graphs. 

In order to treat surfaces with boundary, we employ the usual construction of attaching an additional 2-cell (with function value $\infty$) to each boundary component. This way we obtain a closed surface having the same sequence of order subcomplexes (up to the additional cells) and hence the same persistence pairs as the original surface. 

\subsection{The persistence hierarchy and sequential cancelations}\label{sec:persistenceHierarchy}

Persistence pairs on surfaces carry a certain hierarchical structure that allows us to establish a connection to the cancelation theorem of discrete Morse theory. %
The main result of this section is that persistence pairs on surfaces can always be canceled sequentially if the order of cancelations respects this hierarchy.

\begin{definition*}[parent, child, persistence hierarchy]
On a combinatorial surface $\complex$, let $(\sigma,\tau)$ be a persistence pair with $\dim \sigma=0$, and let $[\sigma]\in H_0(\complex(\sigma))$ be the class created by $\sigma$. Let $\tildesigma$ be the unique cell creating the class $[\tildesigma]\in H_0(\complex(\tau))$ into which $[\sigma]$ gets merged by $\tau$, i.e., $[\tildesigma]\not\in \im(i_*^{\, \tildesigma_-,\, \tau})$ and $[\tildesigma]= i_*^{\, \sigma,\, \tau}([\sigma])$.
Then $\tildesigma$ is called the \emph{parent} of $\sigma$ (in the \emph{persistence hierarchy}), and $\sigma$ is called the \emph{child} of $\tildesigma$. The transitive closure of the child relation is called \emph{descendant}. 
\end{definition*}

Let $(\sigma,\tau)$ and $(\tildesigma,\tildetau)$ be two persistence pairs. If either $\dim\sigma=\dim\tildesigma=0$ and $\tildesigma$ is the parent of~$\sigma$ or $\dim\tau=\dim\tildetau=2$ and $\tildetau^*$ is the parent of $\tau^*$ (with regard to the persistence hierarchy on the dual complex), then we also call the pair $(\tildesigma,\tildetau)$ the \emph{parent} of $(\sigma,\tau)$ and $(\sigma,\tau)$ the \emph{child} of $(\tildesigma,\tildetau)$. The following definition and lemma justify this nomenclature:

\begin{figure}
\center{
\includegraphics[scale=0.9]{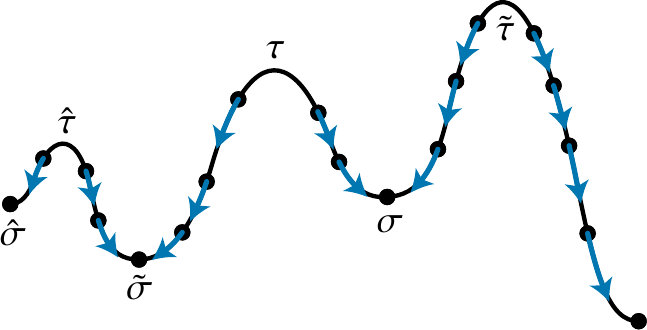}
}
\caption{The persistence hierarchy. Both $(\sigma,\tau)$ and $(\hat\sigma,\hat\tau)$ are children of, and hence nested in, $(\tildesigma,\tildetau)$. Only $(\sigma,\tau)$ needs to be canceled before $(\tildesigma,\tildetau)$ can be canceled.}
\label{fig:hierarchy}
\end{figure}

\begin{definition*}[nested pairs]
On a combinatorial surface $\complex$, let $(\sigma,\tau)$ and $(\tildesigma,\tildetau)$ be two persistence pairs. We say that $(\sigma,\tau)$ is \emph{nested} in $(\tildesigma,\tildetau)$ if\/ $\tildesigma \prec \sigma \prec \tau \prec \tildetau$.
\end{definition*}

\begin{lemma}\label{lem:nestedChild}
Let $(\sigma,\tau)$ be a descendant of $(\tildesigma,\tildetau)$ in the persistence hierarchy. Then $(\sigma,\tau)$ is nested in $(\tildesigma,\tildetau)$.
\end{lemma}

\begin{proof}
Without loss of generality, assume $\dim\sigma=0$; otherwise, by duality, the argument can be applied to $(\tau^*,\sigma^*)$ instead of $(\sigma,\tau)$. %

By definition of the persistence hierarchy, $[\sigma]$ gets merged into the class~$[\tildesigma]\in H_0(\complex(\tau))$ created by $\tildesigma$. This implies that $\tildesigma \prec \sigma$. It also implies that the class created by $\tildesigma$ has not been merged by any cell of $\complex(\tau)$, hence $\tau\prec\tildetau$.
\end{proof}

We now turn our attention to the sequential cancelation of persistence pairs. Note that the cancelation theorem (Theorem~\ref{thm:cancelation}) applies to vector fields, which only provide a \emph{partial order} on the cells, while the notion of persistence is based on a \emph{total order}. After canceling a persistence pair, the new vector field is no longer consistent with the initial total order. It is important to keep in mind that we only talk about persistence pairs of the \emph{initial} total order $\prec$, which is consistent with $(f,V)$; we do not consider a new total order after applying a cancelation (which would complicate things considerably). Applying several cancelations results in a sequence of simplified vector fields:

\begin{definition*}[persistence cancelation sequence]
A \emph{persistence cancelation sequence} is a sequence of gradient vector fields $(V_0,V_1,\ldots,V_n)$ with $V_0=V$, where each $V_i$ is constructed from $V_{i-1}$ by canceling a persistence pair $(\sigma_i,\tau_i)$ using Theorem~\ref{thm:cancelation}.%

A persistence cancelation sequence is called \emph{nested} if in this construction every pair $(\sigma_i,\tau_i)$ nested in another pair $(\sigma_j,\tau_j)$  is canceled first, i.e., $\sigma_j \prec \sigma_i \prec \tau_i \prec \tau_j \Rightarrow i<j$. 

A persistence cancelation sequence is called a \emph{$\delta$-persistence cancelation sequence} if exactly those persistence pairs are canceled that have persistence $\leq\delta$.
\end{definition*}

A persistence pair $(\sigma,\tau)$ can be canceled from a vector field as soon as all descendants have been canceled (compare also to \citet{Edelsbrunner2003Hierarchical} for the existence part of the following statement in a special case):

\begin{lemma}\label{lem:nestedPersistencePairs}
On a combinatorial surface $\complex$, let $(V_0,V_1,\ldots,V_i)$ be a persistence cancelation sequence.
Assume that a persistence pair $(\sigma,\tau)$ has not been canceled in the sequence but that every descendant of ($\sigma,\tau)$ has been canceled. Then there exists a $V_i$-path from~$\partial\tau$ to $\sigma$ and this path is unique. 

Assume further that every persistence pair nested in ($\sigma,\tau)$ has been canceled. If there is a unique $V_i$-path from~$\partial\tau$ to another cell $\tildesigma\neq\sigma$ that is critical in $V_i$ then we have $\sigma\succ\tildesigma$.
\end{lemma}

Figure~\ref{fig:hierarchy} shows that the condition is sufficient but not necessary. The proof of Lemma~\ref{lem:nestedPersistencePairs} relies on a few auxiliary lemmas and is given after Lemma~\ref{lem:uniqueMinimum}.

\begin{lemma}\label{lem:restrictionCritical}
Let $(V_0,V_1,\ldots,V_i)$ be a persistence cancelation sequence and let $(\sigma, \tau)$ be a persistence pair with $\dim \sigma=0$ that has not been canceled in the sequence. Let $\mathcal C$ be the connected component of the subcomplex $\complex(\tau_-)$ containing $\sigma$, %
and let $C$ denote the cells of $\mathcal C$. Then every $(\phi,\rho)\in V_i$ with $\dim\phi=0$ satisfies $\phi \in C \Leftrightarrow \rho \in C$.
\end{lemma}

\begin{proof}
The claim is shown by induction. The base case follows from consistency of the total order~$\prec$ with $(f,V)$.
Consider the cancelation of a persistence pair $(\sigma_i,\tau_i)$. 
If $\dim\sigma_i\neq0$, the tuples in $V_i$ of dimensions $(0,1)$ stay unchanged and the claim immediately follows from the induction hypothesis. %
Now assume $\dim\sigma_i=0$. We show that the claim holds for every $(\phi,\rho)\in V_i\setminus V_{i-1}$.

The cells in $V_i\setminus V_{i-1}$ are $\tau_i$ and the cells on the $V_{i-1}$-path $(\phi_0,\rho_0,\phi_1,\dots,\rho_{r-1},\phi_r)$ from $\phi_0\in\partial\tau_i$ to $\phi_r=\sigma_i$. By the induction hypothesis we have $\phi_k \in C \Leftrightarrow \rho_k \in C$. Because $\mathcal C$ is a subcomplex, we also have $\rho_{k-1} \in C \Rightarrow \phi_k \in C$ (with $\rho_{-1}=\tau_i$). Moreover, if $\sigma_i \in C$ then $\sigma_i$ is a descendant of $\sigma$ and by Lemma~\ref{lem:nestedChild} $(\sigma_i,\tau_i)$ is nested in $(\sigma,\tau)$, implying that $\sigma_i$ and $\tau_i$ are in the same connected component of $\complex(\tau_-)$. Hence we also have $\sigma_i \in C \Rightarrow \tau_i \in C$.
Consequently, either all or none of the cells in $V_i\setminus V_{i-1}$ are contained in $C$ and the claim immediately follows.
\end{proof}

We also require the notion of the restriction of a vector field to a subcomplex:

\begin{definition*}[restriction of a vector field to a subcomplex]
Let $V$ be a discrete vector field on $\complex$ and let $\widetilde\complex$ be a subcomplex of $\complex$ with cells $\widetilde \cells$. The \emph{restriction} of $V$ to $\widetilde\complex$ is $\widetilde V = V \cap \big(\widetilde \cells \times \widetilde \cells\big)$, i.e., the pairs of cells in $V$ that are both in $\widetilde \cells$.
\end{definition*}

As a direct consequence of this definition, the restriction of a vector field $V$ onto a subcomplex may have critical cells that are not critical in $V$:
\begin{lemma}\label{lem:subcomplexCritical}
Let $\widetilde V$ be the restriction of a discrete vector field $V$ on $\complex$ to a subcomplex~$\widetilde\complex$. The critical $d$-cells of $\widetilde V$ are exactly the critical $d$-cells of $V$ that are contained in $\widetilde \complex$ if and only if each pair $(\sigma,\tau)\in V$ with $\dim\sigma=d$ satisfies $\sigma \in \widetilde \cells \Leftrightarrow \tau \in \widetilde \cells$.
\end{lemma}

Moreover, we use the following fact:
\begin{lemma}\label{lem:uniqueMinimum}
Let $V$ be a discrete gradient vector field $V$ on $\complex$ with only one critical 0-cell $\sigma$. Then there is a $V$-path from every $0$-cell $\tildesigma$ to $\sigma$.
\end{lemma}

\begin{proof}
Each $V$-path of dimension 0 ending at a non-critical cell $\tildesigma\neq\sigma$, $(\tildesigma,\tildetau)\in V$, can be extended by $\tildetau$ and the unique 0-cell $\hat\sigma\in\partial\tildetau$, $\hat\sigma\neq\tildesigma$. Since $\cells$ is finite and $V$ does not contain nontrivial closed paths, the extension will eventually end up at~$\sigma$.
\end{proof}

\begin{proof}[Proof of Lemma~\ref{lem:nestedPersistencePairs}]
Without loss of generality, assume $\dim\sigma=0$; otherwise, by duality, the argument can be applied to $(\tau^*,\sigma^*)$ instead of $(\sigma,\tau)$.

Let $\mathcal C$ be the connected component of the subcomplex $\complex(\tau_-)$ created by $\sigma$. Apart from $\sigma$, every 0-cell in $\mathcal C$ that is critical in $V$ is a descendant of $\sigma$. By assumption, all descendants of $\sigma$ have been canceled, and hence $\sigma$ is the only 0-cell in $\mathcal C$ that is critical in $V_i$. 
By Lemmas~\ref{lem:restrictionCritical} and~\ref{lem:subcomplexCritical}, $\sigma$ is also the only critical 0-cell in the restriction of $V_i$ to $\mathcal C$. By Lemma~\ref{lem:uniqueMinimum}, there is a $V_i$-path to $\sigma$ from every 0-cell in $\mathcal C$, in particular from exactly one of the two 0-cells in~$\partial\tau$ since $\partial\tau\cap C$ contains exactly one cell. By Lemma~\ref{lem:uniquePredSucc}, this path is unique.

Now assume that every persistence pair nested in $(\sigma,\tau)$ has been canceled and there is a unique $V_i$-path from~$\partial\tau$ to another cell $\tildesigma\neq\sigma$ that is critical in $V_i$. By assumption, $\tildesigma$ is not a descendant of $\sigma$, meaning that $\tildesigma$ and $\sigma$ are in different connected components of $\complex(\tau_-)$. 
Moreover, $\tildesigma$ creates the component $\widetilde{\mathcal C}\neq\mathcal C$, because otherwise we would have an uncanceled pair $(\tildesigma,\tildetau)$ nested in $(\sigma,\tau)$.
Since $\tau$ is paired with $\sigma$ and merges $\widetilde{\mathcal C}$ and $\mathcal C$, we know that $\sigma$ is a descendant of $\tildesigma$ and $\sigma\succ\tildesigma$.
\end{proof}

As a consequence of Lemma~\ref{lem:nestedPersistencePairs}, we can construct a sequence of cancelations to eliminate all persistence pairs below a certain persistence threshold:

\begin{theorem}\label{thm:simplifiedVectorField}
Let $f$ be a pseudo-Morse function on a combinatorial surface $\complex$ and let~$\delta\geq0$. Then there exists a nested $\delta$-persistence cancelation sequence.
\end{theorem}

\begin{proof}
If the subsequence $(V_0,V_1,\ldots,V_{i-1})$ satisfies the assumptions of~Lemma~\ref{lem:nestedPersistencePairs} for some persistence pair $(\sigma_i,\tau_i)$, we can use Theorem~\ref{thm:cancelation} to construct $V_i$ from $V_{i-1}$.
A canonical choice satisfying these assumptions is given by canceling the persistence pairs $(\sigma_i,\tau_i)$ with persistence $\leq\delta$ according to the order $\prec$ on the negative cells, i.e., $\tau_i \prec \tau_{i+1}$ for every $i$. The claim follows by induction. 
\end{proof}

\subsection{The stability bound}
\citet{CohenSteiner2007Stability} studied properties of \emph{persistence diagrams}, which are a representation of the value pairs $(f(\sigma),f(\tau))$ corresponding to the persistence pairs $(\sigma,\tau)$ of a function $f$. Here we use $\overline \RR=\RR\cup\{-\infty,\infty\}$.
\begin{definition*}[Persistence diagram \citep{CohenSteiner2007Stability}]
The persistence diagram $D(f) \subset \overline \RR^2$ of a pseudo-Morse function~$f$ is the multiset consisting of $(f(\sigma),f(\tau))$ for all persistence pairs~$(\sigma,\tau)$ of $f$, together with all points on the diagonal counted with (countably) infinite multiplicity. An unpaired positive cell $\sigma$ is represented by $(f(\sigma), \infty)$.
\end{definition*} 
The main result of \citep{CohenSteiner2007Stability} is the \emph{Bottleneck Stability Theorem} for persistence diagrams: if two functions are close then their persistence diagrams are also close. Due to the correspondence between piecewise linear functions and discrete pseudo-Morse functions (Section~\ref{sec:PLfunctions}), the statement reads as follows in the language of discrete Morse theory:

\begin{definition*}[Bottleneck distance]
Let $X$ and $Y$ be two multisets of\/ $\overline {\mathbb R}^2$. The \emph {bottleneck distance} is $d_B(X,Y):=\inf_\gamma\sup_{x\in X}\|x-\gamma(x)\|_\infty$, where $\gamma$ ranges over all bijections from $X$ to~$Y$. 
\end{definition*}
Here we assume $(a,\infty)-(b,\infty)=(a-b,0)$, $(a,\infty)-(b,c)=(a-b,\infty)$, and $\|(a,\infty)\|_\infty=\infty$ for $a,b,c\in\RR$.

\begin{theorem}[\citet{CohenSteiner2007Stability}]\label{thm:stability}
Let $f,g:K\to\RR$ be two discrete pseudo-Morse functions. Then the respective persistence diagrams satisfy $d_B(D(f),D(g))\leq\|f-g\|_\infty.$
\end{theorem}
Note that the bottleneck distance provides a metric on the persistence diagrams of pseudo-Morse functions on $\complex$, in particular, $d_B(D(f),D(g))=0$ if and only if $D(f)=D(g)$. Therefore, in contrast to the persistence \emph{pairs}, the persistence \emph{diagram} of a discrete pseudo-Morse function $f$ is well-defined; in particular, it is independent of the total order $\prec$ chosen and even independent of the gradient vector field $V$ consistent with~$f$.
Theorem~\ref{thm:stability} provides a lower bound on the number of persistence pairs among all pseudo-Morse functions $f_\delta$ with $\|f_\delta-f\|_\infty \leq \delta$: %
\begin{corollary}[Stability Bound]\label{cor:criticalPointLowerBound}
For any pseudo-Morse function $f_\delta$ with $\|f_\delta-f\|_\infty \leq\delta$, the number of persistence pairs of $f_\delta$ is bounded from below by the number of persistence pairs of $f$ that have persistence $> 2 \delta$.
\end{corollary}
\begin{proof}
Let $D$ and $D_\delta$ be the persistence diagrams of $f$ and $f_\delta$, respectively. By Theorem~\ref{thm:stability} we have $d_B(D,D_\delta) \leq\delta$. This means that there is a bijection $\gamma$ between $D$ and $D_\delta$ with $\|p-\gamma(p)\|_\infty \leq \delta$ for all $p \in D$. Let $p=(p^*,p^\dagger)=(f(\sigma),f(\tau)) \in D$ represent a persistence pair $(\sigma,\tau)$ of $f$ with persistence $p^\dagger-p^* > 2 \delta$. Letting $q=(q^*,q^\dagger):=\gamma(p)$, this implies that $p^*+\delta \geq q^*$ and $p^\dagger-\delta \leq q^\dagger$. Together with $p^\dagger-p^* > 2 \delta$, this yields $q^\dagger-q^* > 0$. Hence there must be a persistence pair of $f_\delta$ corresponding to each persistence pair of~$f$ with persistence $> 2 \delta$.
\end{proof}

\section{Function simplification guided by discrete gradient vector fields}\label{sec:simplification}

We are interested in functions that achieve the lower bound of Corollary~\ref{cor:criticalPointLowerBound}:

\begin{definition*}[Perfect $\delta$-simplification]
Let $f$ be a pseudo-Morse function on a %
combinatorial surface $\complex$.
A \emph{perfect $\delta$-simplification} of $f$ is a pseudo-Morse function $f_\delta$ such that $\|f_\delta-f\|_\infty \leq \delta$ and the number of persistence pairs of $f_\delta$ 
equals the number of persistence pairs of $f$ that have persistence~$>2 \delta$.

\end{definition*}

In this section, we prove the following central result:
\begin{theorem}\label{thm:2deltasimp}
Let $f$ be a discrete pseudo-Morse function on a combinatorial surface. Then there exists a \emph{perfect $\delta$-simplification} of $f$.
\end{theorem}

The proof of Theorem~\ref{thm:2deltasimp} is constructive and hence leads to an algorithm. The corresponding construction is outlined in Section~\ref{sec:leveling}. Unfortunately, the resulting algorithm has a running time that is quadratic in the input size. We present an \emph{efficient} algorithm in Section~\ref{sec:efficientAlgo}. The proof of its correctness becomes easier once Theorem~\ref{thm:2deltasimp} is established. This is the reason why we present two separate constructions.

\begin{corollary}[Tightness of the stability bound]\label{cor:tightnessPseudoMorse}
Given a discrete pseudo-Morse function $f$ on a surface and $\delta\geq0$, there exists a discrete pseudo-Morse $f_\delta$ consistent with a gradient vector field $V_\delta$ such that $\|f_\delta-f\|_\infty \leq \delta$ and the number of critical points of $V_\delta$ equals the number of critical points of $f$ that have persistence $> 2 \delta$.
\end{corollary}

Using Lemma~\ref{lem:perturbation}, the result can also be stated for (non-degenerate) discrete Morse functions (in a slightly different form, because only critical points with persistence $<2\delta$ can be eliminated within a tolerance of $\delta$ in the set of discrete Morse functions):

\begin{corollary}\label{cor:tightnessMorse}
Given a discrete Morse function $f$ on a surface and $\delta>0$, there exists a discrete Morse function $f_\delta$ such that $\|f_\delta-f\|_\infty < \delta$ and the number of critical points of $f_\delta$ equals the number of critical points of $f$ that have persistence $\geq 2 \delta$.
\end{corollary}

\paragraph{Convention and Notation}
Throughout this section we consider a given pseudo-Morse function $f$ consistent with a gradient vector field $V$ on a combinatorial surface~$\complex$, a strict total order~$\prec$ consistent with~$(f,V)$, and a nested $2\delta$-persistence cancelation sequence $(V_0,\ldots,V_n)$ with $V_0=V$. Moreover, we let $\prec_j\ :=\ \prec_{V_j}$ denote the partial order induced by~$V_j$.

\subsection{The plateau function}\label{sec:leveling}
For every $V_i$ in the cancelation sequence, we inductively define a pseudo-Morse function $f_i$ consistent with $V_i$, see Figure~\ref{fig:cancellation} for an illustration.
By assumption we start with a pseudo-Morse function $f_0:=f$ consistent with $V_0:=V$. Suppose that we have constructed a pseudo-Morse function $f_{i-1}$ consistent with $V_{i-1}$. Let $(\sigma,\tau)$ be the persistence pair that is canceled in the construction of $V_i$ from $V_{i-1}$ using Theorem~\ref{thm:cancelation}.
We define the corresponding \emph{plateau function} $f_i$ as follows:
$$
m_i=\frac{f(\sigma)+f(\tau)}2\quad \text{and} \quad
f_i(\rho):=
\begin{cases}
m_i & \pbox{\textwidth}{$\text{if~} \rho\succeq_{i-1}\sigma \text{~and~} f_{i-1}(\rho) < m_i$ \\
$\text{or~} \rho\preceq_{i-1}\tau \text{~and~} f_{i-1}(\rho) > m_i,$} \\[3\smallskipamount]
f_{i-1}(\rho) & \text{otherwise}.
\end{cases}
$$
This means that the \emph{attracting set} $\{\rho:\rho\succeq_{i-1}\sigma\}$ of $\sigma$ is raised to at least the value $m_i$, and analogously the \emph{repelling set} $\{\rho:\rho\preceq_{i-1}\tau\}$ of $\tau$ is lowered. Hence, $f_i$ creates a local \emph{plateau} at the value $m_i$.
The following lemma is a direct consequence of the way we construct $f_i$ from $f_{i-1}$ and the fact that $f_i$ is constant along the path from $\partial\tau$ to $\sigma$. It can be proven using a straightforward induction argument.

\begin{lemma}\label{lem:consistent}
The plateau function $f_i$ is consistent with both $V_{i-1}$ and $V_i$.
\end{lemma}

Note that the construction of the plateau function does not depend on the properties of combinatorial surfaces but can be applied to regular CW complexes of \emph{arbitrary dimensions}. Moreover, it does not depend on the cancelation persistence pairs: whenever we have a pseudo-Morse function $f$ consistent with a gradient vector field $V$ and $\widetilde V$ is constructed from $V$ by a cancelation using Theorem~\ref{thm:cancelation}, we can obtain a plateau function $\tilde f$ that is consistent with both $V$ and $\widetilde V$.

\subsection{Checking the constraint}\label{sec:constraint}

It remains to show that the plateau construction above is admissible, i.e., that all of the functions $f_i$ satisfy the $\delta$-constraint. 
\begin{lemma}\label{lem:deltaconstraint}
Each plateau function $f_i$ satisfies $\| f_i - f \|_\infty \leq \delta$.
\end{lemma}

\begin{proof}

We show the statement by induction.
The base case is trivial since $f_0=f$.

Let $(\sigma,\tau)$ be the persistence pair that is canceled when constructing $V_i$ from $V_{i-1}$. We show that
the $\delta$-constraint is neither violated by increasing the value of any cell $\rho$ in the attracting set of $\sigma$ in $V_{i-1}$, nor by decreasing the value of any cell in the repelling set of $\tau$. Since $f_i(\rho) = f_{i-1}(\rho)$ for all cells $\rho$ not treated in these two cases, the claim follows.

We first show 
$| f_i(\rho) - f(\rho) | \leq \delta$ for any cell $\rho \succeq_{i-1}\sigma$ with $f_{i-1}(\rho)<m_i$. 
By the induction hypothesis we have a lower bound $f_{i-1}(\rho) \geq f(\rho)-\delta$.
By construction of $f_i$, the value of $\rho$ is increased: $f_i(\rho) = m_i > f_{i-1}(\rho)$. Therefore, the lower bound remains valid after step $i$:
$$
  f_i(\rho) > f_{i-1}(\rho) \geq f(\rho) - \delta .
$$
To show the upper bound $f_i(\rho)\leq f(\rho) +\delta$, we first use $f(\tau)-f(\sigma) \leq 2\delta$ to obtain
$$
  f_i(\rho) = m_i = \frac{f(\sigma)+f(\tau)}2 
  \leq \frac{f(\sigma)+(f(\sigma)+2\delta)}2 
  = f(\sigma) +\delta.
$$
This is almost the desired inequality except that the right hand side contains $f(\sigma)$ instead of $f(\rho)$. To finish the proof, it therefore suffices to show that $f(\sigma)\leq f(\rho)$. This, in turn, is a consequence of the facts that, according to Lemma~\ref{lem:upperSetValues}, $\sigma \prec_{i-1} \rho$ implies $\sigma \prec \rho $, and that $\prec$ is consistent with $(f,V)$.

It remains to show that~$| f_i(\rho) - f(\rho) | \leq \delta$ for any cell~$\rho \preceq_{i-1}\tau$ with $f_{i-1}(\rho)>m_i$. The proof of this statement is analogous to the above.
\end{proof}

Before proving Lemma~\ref{lem:upperSetValues}, we first investigate how the reversal of a gradient vector field may change the induced partial order (see Figure~\ref{fig:reversing} for an example):

\begin{figure}
\center{
\includegraphics[scale=0.9]{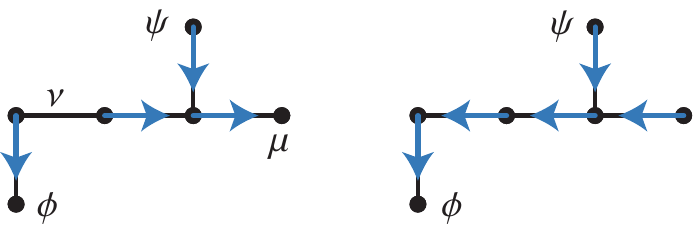}
}

\caption{Example illustrating Lemma~\ref{lem:changedRelation}. Left: gradient vector field $W$ (before reversing the path from $\partial\nu$ to $\mu$). Right: gradient vector field $\widetilde W$ (after path reversal). Note that we have the new relation $\phi\prec_{{\widetilde W}}\psi$ (corresponding in this example to a $\widetilde W$-path from $\psi$ to $\phi$). In the example, the conclusion $\phi\preceq_W\nu$ and $\mu\preceq_W\psi$ of Lemma~\ref{lem:changedRelation} is reflected by the two $W$-paths from $\partial\nu$ to $\phi$ and from $\psi$ to $\mu$, respectively.
}
\label{fig:reversing}
\end{figure}

\begin{lemma}\label{lem:changedRelation}
Let $\mu,\nu,\phi,\psi$ be (not necessarily disjont) cells of a regular CW complex $\complex$, and let $W$ and~${\widetilde W}$ be two gradient vector fields. Assume that the cells $\mu,\nu$ are critical in $W$ and that ${\widetilde W}$ is constructed by reversing~$W$ along the unique $W$-path from $\partial\nu$ to $\mu$. Assume further that $\phi\not\prec_W\psi$ and $\phi\prec_{{\widetilde W}}\psi$. Then $\phi\preceq_W\nu$ and $\mu\preceq_W\psi$.
\end{lemma}

\begin{proof}
By definition of the induced partial order, $\phi\prec_{{\widetilde W}}\psi$ implies that there exists a sequence $(\rho_1,\dots,\rho_k)$ with $\rho_1=\phi$, $\rho_k=\psi$ and $\rho_i\gets_{\widetilde W}\rho_{i+1}$ for all $1\leq i\leq k-1$. Here either $\rho_i$~is a facet of $\rho_{i+1}$ or $\rho_{i+1}$~is a facet of $\rho_i$, and we therefore also have either $\rho_i\gets_W\rho_{i+1}$ or~$\rho_i\to_W\rho_{i+1}$. But since $\phi\not\prec_W\psi$, there exists a smallest index~$j$ such that~$\rho_j\to_W\rho_{j+1}$. Since the relations $\gets_W$ and $\gets_{\widetilde W}$ differ only along the $W$-path from $\partial\nu$ to $\mu$ (including~$\nu$), it follows that the cells $\rho_j$ and $\rho_{j+1}$ are contained in this $W$-path. Hence we have~$\rho_j\preceq_W\nu$. Moreover, by the choice of $j$ we have 
$\phi=\rho_1\preceq_W\rho_j$.
 Therefore we conclude that $\phi\preceq_W\nu$. By an analogous argument one also shows that~$\mu\preceq_W\psi$.
\end{proof}

\begin{lemma}\label{lem:upperSetValues}

Let 
$(V_0,\ldots,V_n)$ be a nested persistence cancelation sequence and let
$(\sigma,\tau)$ be a persistence pair of $\prec$ with $\sigma$ and $\tau$ critical cells of $V_i$. Then for any $\rho \in \cells$,
\begin{compactenum}[(a)]
   \item $\rho \succ_i \sigma$~\,implies~\,$\rho \succ \sigma$, and
   \item $\rho \prec_i \tau$~\,implies~\,$\rho \prec \tau$.
\end{compactenum}
\end{lemma}

\begin{proof}
We only present the proof of part (a), which is done again by induction: we show that $\rho \succ_{i} \sigma$ implies $\rho \succ \sigma$ for all $0\leq i\leq n$. Part (b) can be shown analogously. 

The base case $i=0$ is trivial since $\succ$ is a linear extension of $\succ_0$.
Assume that $\rho \succ_{i} \sigma$. If $\rho \succ_{i-1} \sigma$, then the claim follows directly from the induction hypothesis. Hence we assume that $\rho \not \succ_{i-1} \sigma$. Let $(\tildesigma,\tildetau)$ be the persistence pair that is canceled when constructing $V_i$ from $V_{i-1}$; this implies $\tildesigma\prec_{i-1}\tildetau$. From Lemma~\ref{lem:changedRelation} with $(W,\widetilde W)=(V_{i-1},V_i)$ and $(\mu,\nu,\phi,\psi)=(\tildesigma,\tildetau,\sigma,\rho)$, we infer that $\sigma\preceq_{i-1}\tildetau$ and $\tildesigma \preceq_{i-1} \rho$. This has two consequences:
\begin{compactenum}[(i)]
\item $\sigma\prec_{i-1}\tildetau$ ~(since $\sigma$ is critical in $V_i$ while $\tildetau$ is not), and
\item $\tildesigma \preceq \rho$ ~(by the induction hypothesis).
\end{compactenum}
To finish the proof of the claim, by (ii) it suffices to show that $\sigma \prec \tildesigma$.
We proceed by case analysis on the dimensions of $\tildesigma$ and $\sigma$. Since these two cells are positive by assumption, they have dimension less than 2.

\textbf{Case 1} $(\dim\sigma=1$, $\dim\tildesigma=0)$: This case cannot occur since reversing the $V_{i-1}$-path from the 1-cell $\tildetau$ to the 0-cell $\tildesigma$ does not change the attracting set of any critical 1-cell (and in particular $\sigma$), contradicting $\rho \not \succ_{i-1} \sigma$ and $\rho \succ_{i} \sigma$.

\textbf{Case 2} $(\dim\sigma=0$, $\dim\tildesigma=1)$: First assume $\tau \prec \tildetau$. If additionally $\tildesigma \prec \sigma$, this contradicts the assumption that the cancelation sequence is nested  and $(\sigma,\tau)$ is canceled after ($\tildesigma,\tildetau)$. Therefore $\tau \prec \tildetau$ implies $\sigma \prec \tildesigma$.

Now assume $\tau \succ \tildetau$. This means that $\sigma$ creates a connected component that is not yet merged in $\complex(\tildetau)$. Since $\sigma\prec_{i-1}\tildetau$ by (i), there is a sequence $(\rho_1,\dots,\rho_k)$ with $\rho_1=\sigma$, $\rho_k=\tildetau$, and $\rho_j\gets_{V_{i-1}}\rho_{j+1}$ for all $1\leq j\leq k-1$. For each $\rho_j$ we trivially have $\rho_j\prec_{i-1}\tildetau$ and hence $\rho_j\prec\tildetau$ by the induction hypothesis, implying that $\rho_j\in\complex(\tildetau)$.
Moreover, since either $\rho_j$~is a facet of $\rho_{j+1}$ or $\rho_{j+1}$~is a facet of $\rho_j$, we know that all $\rho_j$, and in particular $\sigma$ and $\tildetau$, are in the same connected component of~$\complex(\tildetau)$. In an analogous way one shows that $\tildesigma$ and $\tildetau$, and hence $\sigma$ and $\tildesigma$, are in one and the same connected component. Since we know that $\sigma$ created that component, it follows that $\sigma \prec \tildesigma$.

\textbf{Case 3} $(\dim\sigma=\dim\tildesigma \in \{0,1\})$:
The relation $\sigma\prec_{i-1}\tildetau$ from (i) above implies the existence of a $V_{i-1}$-path from $\tildetau$ to $\sigma$. 
We will show by contradiction that this path must be unique. To see this, assume that there are two $V_{i-1}$-paths from $\tildetau$ to $\sigma$. 
Without loss of generality, assume that $\dim\sigma=\dim\tildesigma=0$ (and hence $\dim\tildetau=1$); otherwise, by duality the following argument can be applied to $\sigma^*,\tildetau^*$ instead of $\tildetau,\sigma$. By Corollary~\ref{cor:twoPaths}, each of the 0-cells in $\partial\tildetau$ must belong to exactly one of the two $V_{i-1}$-paths from $\tildetau$ to $\sigma$. Now by a similar argument as in Case 2 above, we obtain that each cell of these two $V_{i-1}$-paths is contained in the same connected component of $\complex(\tildetau_-)$ as $\sigma$. But since $\tildetau$ is a negative 1-cell, the two 0-cells in its boundary belong to different connected components of $\complex(\tildetau_-)$, a contradiction.

Hence, there is a unique $V_{i-1}$-path from $\tildetau$ to $\sigma$. 
Lemma~\ref{lem:nestedPersistencePairs} asserts that $\tildesigma$ is the largest cell (with regard to $\prec$) with a unique $V_{i-1}$-path from $\tildetau$ to $\tildesigma$. Since $\sigma\neq\tildesigma$, we obtain $\sigma \prec \tildesigma$.
\end{proof}

\begin{proof}[Proof of Theorem~\ref{thm:2deltasimp}]
According to Theorem~\ref{thm:simplifiedVectorField} there exists a nested $2\delta$-persistence cancelation sequence $(V_0,V_1,\ldots,V_n)$ for the pseudo-Morse function $f$.
Let $f_n$ be the plateau function corresponding to $V_n$. Since $f_n$ is consistent with $V_n$ by Lemma~\ref{lem:consistent} and $\| f_n-f\|_\infty \leq \delta$ by Lemma~\ref{lem:deltaconstraint}, it is a perfect $\delta$-simplification.
\end{proof}

\section{An efficient algorithm}\label{sec:efficientAlgo}

The definition of the plateau function in the previous section canonically leads to an algorithm that runs in time quadratic in the input size. In this section we present a method for computing a perfect $\delta$-simplification in time dominated by the computation of persistence pairs, i.e., $\mathcal O(\mathop\mathrm{sort}(n))$, where $n=|\cells|$ is the number of cells of $\complex$. Apart from this computation, all steps of our algorithm take linear time $\mathcal O(n)$. We stress that pre- and post-processing steps, like conversion from and to PL functions, also require only linear time~$\mathcal O(n)$.

The algorithm can be summarized as follows. First, persistence pairs are computed using a variant of Kruskal's algorithm for minimum spanning trees. Next, the persistence pairs are used to construct a simplified gradient vector field by a graph traversal of both the primal and dual 1-skeleton. In a third step, the simplified vector is used to compute the simplified function by a graph traversal on the Hasse diagram of the partial order induced by the simplified vector field.

\subsection{Defining a consistent total order}\label{sec:totalOrder}
Assume we are given a pseudo-Morse function $f$ consistent with a discrete gradient vector field $V$ as input. We write $\phi\simeq_V\rho$ if neither $\phi\prec_V\rho$ nor $\phi\succ_V\rho$, and similarly for~$\simeq_f$. Let $\prec_T$ be an arbitrary total order on $\cells$. We define the order $\prec$ as the lexicographic order given by $\prec_f$, $\prec_V$, and $\prec_T$: we have $\phi\prec\rho$ if and only if either 
\begin{compactenum}[(a)]
\item
$\phi\prec_f\rho$, 
\item $\phi\simeq_f\rho$ and $\phi\prec_V\rho$, or 
\item $\phi\simeq_f\rho$ and $\phi\simeq_V\rho$ and $\phi\prec_T\rho$.
\end{compactenum}
Now assume that $f$ is constructed from data given as a PL or piecewise constant function as explained in Section~\ref{sec:PLfunctions}. Then $V$ is the empty vector field (all cells are critical), meaning that $\phi\prec_V\rho$ if and only if $\phi$ is a face of $\tau$.
If now the order $\prec_T$ is chosen such that the cells are sorted by dimension, then $\phi\prec_V\rho$ implies $\phi\prec_T\rho$. The definition now simplifies to: $\phi\prec\rho$ if and only if either 
\begin{compactenum}[(a)]
\item
$\phi\prec_f\rho$ or 
\item $\phi\simeq_f\rho$ and $\phi\prec_T\rho$.
\end{compactenum}

\subsection{Computing persistence pairs}
Recall that the persistence pairs of dimension $(0,1)$ are determined solely by the 1-skeleton $G$ of $\complex$.
Therefore, persistence pairs can be computed by applying a variant of Kruskal's algorithm \citep{Kruskal1956Shortest} for finding a minimum spanning tree to both the primal and the dual 1-skeleton \citep{Edelsbrunner2002Topological,Attali2009Simplification}. Let $G$ be the 1-skeleton of $\complex$ and $M(G)$ the minimum spanning tree of $G$ (using the total order~$\prec$ for determining the edge weights, which implies uniqueness of $M(G)$). 
Kruskal's algorithm for computing $M(G)$ initializes a graph $T$ with the vertices of $G$, sweeps over the edges of $G$ in order~$\prec$, adds to $T$ every edge of $G$ that does not create a 1-cycle, and returns the final graph $T$. Note that the set of edges of $M(G)$ consists of all negative 1-cells together with all 1-cells $\tau$ with $(\sigma,\tau)\in V$ for some $\sigma$; all other 1-cells create a cycle in $T$.
When encountering a negative 1-cell, we compute the persistence of the corresponding $(0,1)$ pair by storing for each connected component of the intermediate graph $T$ the 0-cell that created it. 
Clearly we obtain all dimension $(0,1)$ persistence pairs this way. 
Simultaneously, we construct the subgraph $M_\delta(G)$ of $M(G)$ not containing the negative 1-cells with persistence~$> 2\delta$. 
In an analogous way, for the dual 1-skeleton $G^*$ we can compute the minimum spanning tree $M(G^*)$ and obtain the subgraph $M_\delta(G^*)$ together with all $(1,2)$ persistence pairs.

Kruskal's algorithm has a time complexity of $\mathcal O(\mathop\mathrm{sort}(n))$, yielding a complexity of $\mathcal O(n\log n)$ for comparison-based sorting. Assuming that the function values are represented by a small $(\mathcal O(\log n))$ word size, \citet{Attali2009Simplification} point out that persistence pairs on a graph can be computed in linear time $\mathcal O(n)$ on a RAM using radix sort together with a linear-time algorithm for minimum spanning trees.

\subsection{Extracting the gradient vector field}\label{sec:efficientAlgoVF}
We now explain how to construct a simplified gradient vector field $V_\delta$. To this end, we traverse (using depth-first search) each of the connected components of the primal graph $M_\delta(G)$ (constructed in the previous section) from the 0-cell that created the component. During this traversal, whenever we encounter an edge (1-cell) $\psi$ that connects a previously visited vertex (0-cell) $\rho$ with an unvisited vertex $\phi$, we add $(\phi, \psi)$ to the gradient vector field~$V_\delta$. This construction takes $\mathcal O(n)$ time.

We perform an analogous traversal for the dual graph $M_\delta(G^*)$. Again, whenever we encounter an edge $\psi^*$ that connects a visited vertex $\rho^*$ with an unvisited vertex $\phi^*$ (with $\psi$ a 1-cell and $\rho,\phi$ 2-cells of the original complex), we add $(\psi,\phi)$ to the gradient vector field $V_\delta$. Note that the final $V_\delta$ results from both the primal and dual traversals and is a vector field on $\complex$.

\begin{theorem}\label{thm:efficientVectorField}
The gradient vector field $V_\delta$ is identical to the final vector field $V_n$ of a $2\delta$-persistence cancelation sequence $(V_0,\dots,V_n)$.
\end{theorem}

\begin{proof}
First observe that if $(\sigma,\tau)\in V_n$ and $\dim\sigma=0$, then both $\sigma$ and $\tau$ are cells of $M_\delta(G)$ since all non-critical cells of $V_n$ either are non-critical in $V$ as well or have persistence $\leq2\delta$ (with respect to~$f$ and~$\prec$).
Moreover, the 0-cells creating a connected component of $M_\delta(G)$ are the only critical 0-cells of $V_n$ (by definition) and of $V_\delta$ (by construction). Since $M_\delta(G)$ is a tree, the pairs $(\sigma,\tau)\in V_n$ with $\dim\sigma=0$ are uniquely defined by this property.
By applying the dual argument to $M_\delta(G^*)$, the statement follows.
\end{proof}

\subsection{Constructing the simplified function}\label{sec:definingFucntionEfficient}
Finally, we construct a function $f_\delta$ (different from the plateau function defined in Section~\ref{sec:leveling}) that is consistent with the simplified gradient vector field $V_\delta$. Consider the Hasse diagram $H:=H_{V_\delta}$ of the strict partial order $\prec_{V_\delta}$ as described in Section~\ref{sec:pseudoMorse}. We visit the vertices $\cells$ of $H$ in a linear extension of $\prec_{V_\delta}$. The problem of finding a linear extension of a partial order is also called \emph{topological sorting} and can be solved using depth-first search on $H$ \citep{Cormen2009Introduction}%
. At each visited cell $\sigma$, we define $f_\delta(\sigma)$ as the minimum value that satisfies the lower bound $f_\delta(\sigma)\geq f(\sigma)-\delta$ \emph{and} renders $f_\delta$ consistent with $V_\delta$, i.e.,
$$
f_\delta(\sigma)=\max\Big(f(\sigma)-\delta,\max_{\rho \gets_{V_\delta}\sigma}f_\delta(\rho)\Big).
$$
The construction of $f_\delta$ also takes $O(n)$ time.

\subsection{Correctness of the algorithm}\label{sec:correctness}

\begin{theorem}
The function $f_\delta$ constructed using the above algorithm is a perfect $\delta$-simplification of $f$.
\end{theorem}

\begin{proof}
By construction $f_\delta$ is consistent with $V_\delta$. At the same time, by Theorem~\ref{thm:efficientVectorField}, $V_\delta$ is the final vector field of a $2\delta$-persistence cancelation sequence. Therefore, by the definition of a perfect $\delta$-simplification, it only remains to show that the constraint $\| f_\delta-f\|_\infty \leq \delta$ is satisfied. The lower bound $f_\delta\geq f-\delta$ is satisfied by construction. It remains to show the upper bound $f_\delta\leq f+\delta$.

Observe that the set of all perfect $\delta$-simplifications consistent with $V_\delta$ is defined by a set of linear inequalities: the upper and lower bounds on the function values given by $f\pm\delta$, and the inequalities that define consistency with $V_\delta$. Therefore, the set of $\delta$-simplifications is a convex polyhedron $P\subset\RR^{n}$ with $n=|\cells|$. The polyhedron $P$ is bounded since it is a subset of the product of intervals $\prod_{\sigma\in \cells}[f(\sigma)-\delta,f(\sigma)+\delta]$. From Theorem~\ref{thm:2deltasimp}, we know that $P$ is not empty. We now show that $f_\delta$ is contained in~$P$.

First, consider the (unbounded) convex polyhedron $\tilde P$ defined by the lower bound $f_\delta\geq f-\delta$ and the inequalities induced by $V_\delta$%
. By construction, $f_\delta$ is contained in $\tilde P$. Moreover, again by construction,  $f_\delta$ minimizes the function value of any cell among all functions in $\tilde P$. In other words, for any function $\tilde f$ in $P\subset\tilde P$, we have $\tilde f \geq f_\delta$. This implies the upper bound $f_\delta\leq f+\delta$.
\end{proof}

\section{Discussion}\label{sec:discussion}

\subsection{Computational results}

We implemented the algorithm of Section~\ref{sec:efficientAlgo} in C++. For a complex with over 4 million cells (the cubical complex for a $1025\times1025$ pixel image), we obtained a running time of about 15 seconds for computing a perfect $\delta$-simplification on a 2.4GHz Intel Core 2 Duo laptop.

\subsection{Symmetrizing the algorithm}\label{sec:symmetrizing}
The method described in Section~\ref{sec:efficientAlgo} assigns to each cell the smallest possible value. As a consequence, the output function differs from the input function $f$ even if the input function is already a perfect $\delta$-simplification. %
Moreover, the method is not symmetric in the sense that we obtain an output function which \emph{maximizes} the values if we apply the algorithm to the function $-f$ on the dual complex and return the negative of the simplified function. Since both the minimal and maximal solutions are points of a convex polyhedron as explained in Section~\ref{sec:correctness}, we can take the component-wise arithmetic mean to obtain another perfect $\delta$-simplification. 

With this modification, if the input function $f$ is already a perfect $\delta$-simpli\-fication, then the minimal solution is given by $f-\delta$, while the maximal solution equals $f+\delta$, so the arithmetic mean of both solutions returns $f$ again as desired.

\subsection{Flooding and carving artifacts}
\begin{figure*}
\center{
\includegraphics[width=0.3\textwidth]{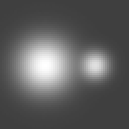}\hfill
\includegraphics[width=0.3\textwidth]{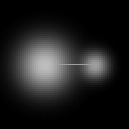}\hfill
\includegraphics[width=0.3\textwidth]{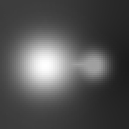}
}
\caption{Visualization of simplification artifacts. Function values indicated by gray levels. Left: Original function. Middle: Function obtained by the algorithm of Section~\ref{sec:efficientAlgo}. Note the bright path joining the two spots. Right: Function obtained after constraint energy minimization according to Section~\ref{sec:energyMethods}. While the simplified topological structure is maintained, the visual appearance is closer to the original function.}
\label{fig:artifacts}
\end{figure*}

\begin{figure*}
\center{
\includegraphics[width=0.715\textwidth]{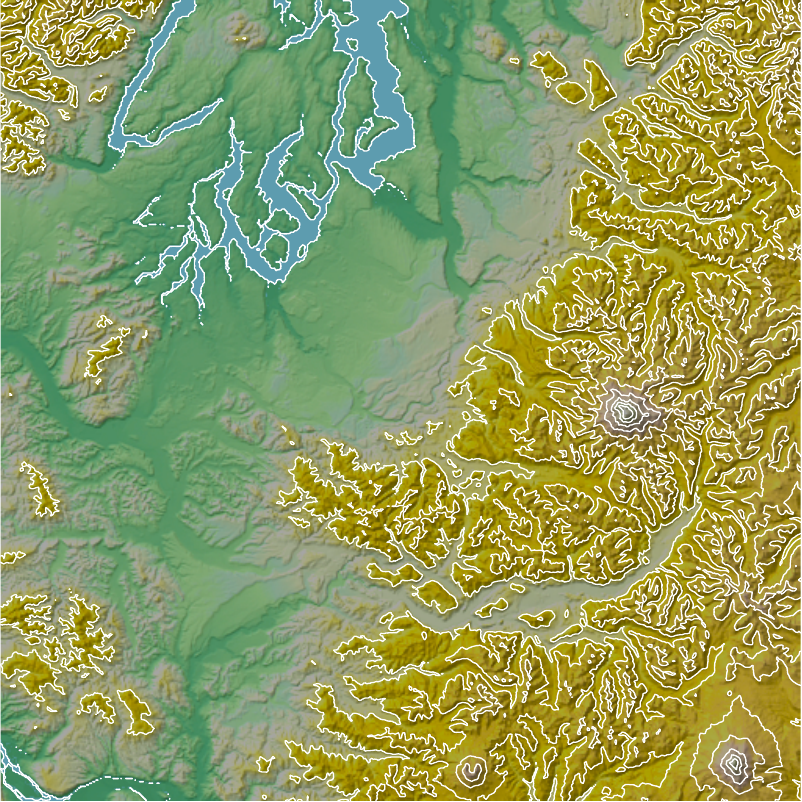%
}\hfill
\includegraphics[scale=1]{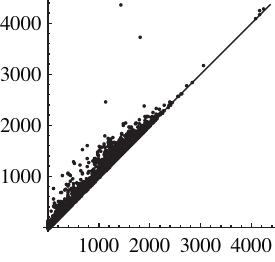%
}%
}
\center{
\includegraphics[width=0.715\textwidth]{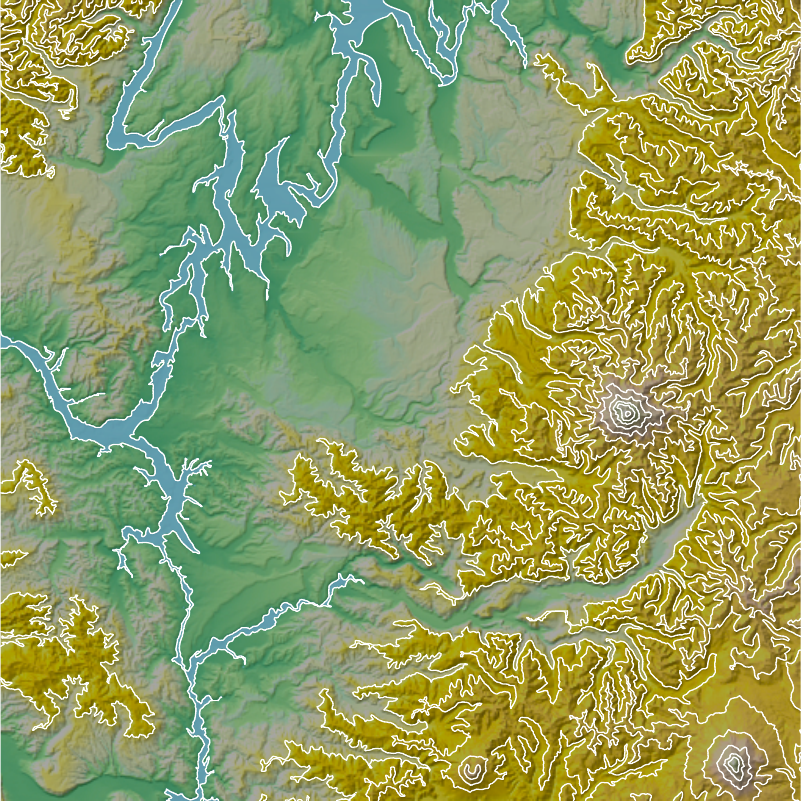%
}\hfill
\includegraphics[scale=1]{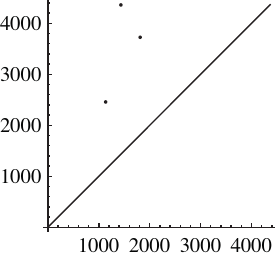%
}%
}
\caption{Top: Topographic map of elevation data set ``Puget Sound'' \cite{LargeGeometricModels}, showing the region around Tacoma. %
Contour lines shown every 500 meters. Elevation data is converted from a $512\times512$ grid into a pseudo-Morse function on $1050625$ cells. $33120$ critical cells have persistence $>0$ (persistence diagram shown on the right). Bottom: Simplified elevation function obtained after constraint energy minimization according to Section~\ref{sec:energyMethods} with $\delta=500$ meters. The 
function has 1 minimum, 3 saddles, and 3 maxima.}
\label{fig:puget}
\end{figure*}

Since the methods presented in the present article can be seen as combinations of the \emph{carving} and \emph{flooding} approaches, they also inherit some characteristics that may not always be desirable in practical applications (see Figure~\ref{fig:artifacts}).

Carving methods \citep{Soille2004Morphological,Edelsbrunner2006PersistenceSensitive,Attali2009Simplification}) cancel a pair of critical cells by changing only the repelling or attracting set of the 1-cell (saddle). This results in a noticeable thin path being carved in the function. On the other hand, modifying extrema, i.e., lowering maxima and raising minima, produces regions with constant function value; this is called \emph{filling} or \emph{flooding} \citep{Jenson1988Extracting,Danner2007TerraStream}). Although this effect is less disturbing, it might appear unnatural in certain applications.
In the next section, we propose a way to remedy both kinds of artifacts.

\subsection{Combining topological simplification and energy methods}\label{sec:energyMethods}

As mentioned in Section~\ref{sec:correctness}, the set of perfect $\delta$-simplifications consistent with the simplified gradient vector field $V_\delta$ is a convex polyhedron $P$. Hence, the presented method can be combined with energy minimization methods, since the polyhedron $P$ can be used as the feasible region for an arbitrary convex optimization problem. For example, we used the interior point solver Ipopt~\citep{Waechter2006Ipopt} to minimize (a discretization of) the Dirichlet energy of the difference $f_\delta-f$ in order to obtain a function $f_\delta$ that looks as similar as possible to the input function $f$ (see Figures~\ref{fig:artifacts} and~\ref{fig:puget}). Alternatively, we minimized the Dirichlet energy of the simplified function itself in order to obtain smooth contour lines.

\subsection{A counterexample for general 2-complexes}\label{sec:counterexample}

The example of Figure~\ref{fig:counterexample} shows that a perfect $\delta$-simplification may not exist on a non-manifold 2-dimen\-sional cell complex. For the sake of simplicity, the example is given for a non-regular CW complex; it is straightforward to rephrase this example using a regular CW complex by subdividing the cells. 
The complex consists of two 0-cells $\zeta$ and $\gamma$ with $f(\zeta)=f(\gamma)=0$, three 1-cells $a$, $b$, and $c$ with $f(a)=1$, $f(b)=2$, and $f(c)=0$, and the two 2-cells $A$ and $B$ with $f(A)=2$ and $f(B)=3$. Note that the complex is not manifold since it is not locally euclidean at the 1-cell $b$. The persistence pairs are $(a,A)$, $(b,B)$, and $(\gamma,c)$. To obtain a perfect $\delta$-simplification for $\delta=0.5$, one would need to set $f_\delta(b)=f_\delta(B)=2.5$ and $f_\delta(a)=f_\delta(A)=1.5$. The corresponding simplified gradient vector field would be $V_\delta=\{(a,A),(b,B)\}$. But since $b$ is a facet of $A$, we must have $f_\delta(b)\leq f_\delta(A)$. Hence, we cannot cancel both $(a,A)$ and $(b,B)$ at the same time. This constellation also appears in~\citep{Edelsbrunner2002Topological} under the name \emph{conflict of type (1,2)}.

\begin{figure}[h]
\center{
\includegraphics[scale=1]{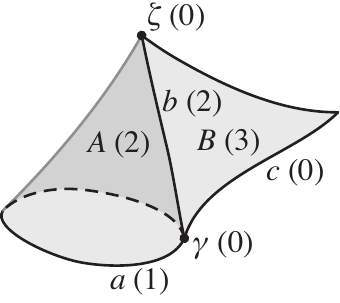}
}
\caption{A discrete Morse function on a 2-complex that does not have a perfect $\delta$-simplification. The function values of the cells are indicated in brackets.}
\label{fig:counterexample}
\end{figure}

Since such a 2-complex can also appear as a level subcomplex of an $n$-manifold CW complex for $n\geq3$, the example also shows that a perfect $\delta$-simplification does not always exists for functions on manifolds.

\subsection{Removing local extrema from functions on manifolds}
As a concluding remark, we want to mention that the same constructions and proofs presented in this article can also be adapted to the problem of minimizing the number of local extrema of a pseudo-Morse function within a $\delta$-tolerance on any $d$-dimensional manifold CW complex. 
\begin{problem*}[Extrema simplification on manifolds]
Given a pseudo-Morse function $f$ on a regular manifold CW complex and a real number $\delta\geq0$, find a function $f_\delta$ subject to $\|f_\delta-f\|_\infty \leq \delta$ such that $f_\delta$ has a minimum number of local extrema. 
\end{problem*}
\begin{theorem}
Given a pseudo-Morse function $f$ on a finite regular closed manifold CW complex and a real number $\delta\geq0$, there exists a pseudo-Morse function $f_\delta$ such that $\|f_\delta-f\|_\infty \leq \delta$ and the number of local extrema of $f_\delta$ equals the number of local extrema of $f$ that have persistence $>2 \delta$. This number is minimal.
\end{theorem}

Note that in the case $d=2$ this problem is equivalent to the topological simplification problem by the following argument. Let $c_i$ denote the number of critical cells of dimension $i$. Since the Euler characteristic $\chi=c_0-c_1+c_2$ is a topological invariant and we have $c_0+c_1+c_2=2(c_0+c_2)-\chi$, the number of critical points is minimal if and only if the number of extrema is minimal.

\paragraph*{Acknowledgements}
The ``Puget Sound'' data set used in Figure~\ref{fig:puget} is taken from the Large Geometric Models Archive of the Georgia Insitute of Technology. The original elevation data is obtained from The United States Geological Survey (USGS), made available by The University of Washington.

\bibliography{baueru,simplification}

\end{document}